\documentclass[11pt,a4paper]{amsart}
\usepackage[dvips,colorlinks=true,linkcolor=blue]{hyperref}
\usepackage{subfigure,caption}
\usepackage{amssymb}
\usepackage{graphicx}
\usepackage[rflt]{floatflt}
\numberwithin{equation}{section}
\newcommand{\smallpagebreak}{{\par\vspace{2 mm}\noindent}}

\newcommand{\dsize}{\textstyle}
\newcommand{\D}{\displaystyle}

\newcommand{\R}{{\mathbb R}}

\newcommand{\Z}{{\mathbb Z}}
\newcommand{\N}{{\mathbb N}}
\newcommand{\C}{{\mathbb C}}
\newcommand{\Q}{{\mathbb Q}}

\newcommand{\Fc}{{\mathcal F}}
\newcommand{\Gc}{{\mathcal G}}
\renewcommand{\mod}{{\rm mod}\,}

\newcommand{\vers}{\operatornamewithlimits{\to}}
\newcommand{\equ}{\operatornamewithlimits{\sim}}
\newcommand{\car}{\bold{1}}
\theoremstyle{plain}
\newtheorem{Th}{Theorem}[section]
\newtheorem{Le}{Lemma}[section]
\newtheorem{Pro}{Proposition}[section]
\newtheorem{Cor}{Corollary}[section]
\theoremstyle{definition}
\newtheorem{Rem}{Remark}[section]

\title{An exact renormalization formula for Gaussian exponential sums
  and applications}
\author{Alexander Fedotov} \author{Fr{\'e}d{\'e}ric Klopp}
\address[Alexander Fedotov]{Departement of Mathematical Physics, St
  Petersburg State University, 1, Ulia\-novskaja, 198904 St
  Petersburg-Petrodvorets, Russia}
\email{\href{mailto:fedotov.s@mail.ru}{fedotov.s@mail.ru}}
\address[Fr{\'e}d{\'e}ric Klopp]{LAGA, Institut Galil{\'e}e, CNRS UMR
  7539, Universit{\'e} de Paris-Nord, Avenue J.-B.  Cl{\'e}ment, F-93430
  Villetaneuse, France\ et \ Institut Universitaire de France}
\email{\href{mailto:klopp@math.univ-paris13.fr}{klopp@math.univ-paris13.fr}}
\thanks{The authors were supported by the grant CNRS PICS 4224/RFBR
  07-01-92169}
\begin{document}

\begin{abstract}
  In the present paper, we derive a renormalization formula ``{\`a} la
  Hardy-Littlewood'' for the Gaussian exponential sums with an exact
  formula for the remainder term. We use this formula to describe the
  typical growth of the Gaussian exponential sums.
  \vskip.5cm
  \par\noindent \textsc{R{\'e}sum{\'e}.}
  Dans cet article, nous obtenons une formule de renormalisation ``{\`a}
  la Hardy-Littlewood'' pour des sommes exponentielles gaussiennes
  avec une formule exacte pour les termes de reste.  Nous utilisons
  cette formule pour d{\'e}crire la croissance typique de ces sommes.
\end{abstract}
\maketitle
Let $(a,b)\in(0,1)\times(-1/2,1/2]$ and consider the Gaussian
exponential sum
\begin{equation}
  \label{eq:SN}
  S(N,a,b)=\sum_{0\leq n\leq N-1} 
  e\left(-\frac{a n^2}2+nb\right), \quad N=1,2,3\dots.   
\end{equation}
where $e(z)=e^{2\pi iz}$. We set $S(0,a,b)=0$.\\[2mm]
Such sums have been the object of many studies (see
e.g.~\cite{MR1555099,MR0563894,MR714822,MR1682276,MR1995685}) and have
applications in various fields of mathematics and physics. In the
present paper we prove a renormalization formula (see
Theorem~\ref{th:S2}) analogous to the one first introduced
in~\cite{MR1555099}. In our formula the ``remainder term'' is given
explicitly by a special function (see section~\ref{sec:Fc}). We use
this renormalization formula to obtain results on the typical growth
and on the graphs of the exponential sums (\ref{eq:SN}) (see
Figure~\ref{fig:fresnel}).\\
Let us now present our main results on the growth of $S(N,a,b)$. To
our knowledge, up to the present work, the growth was studied mainly
in the case $b=0$ (\cite{MR0563894,MR1682276}). As we shall see, the
nontrivial $b$ does change the rate of growth. We prove
\begin{Th}
  \label{thr:2}
  Let $g:\R_+\to\R_+$ be a non increasing function. Then, for almost
  every $(a,b)\in (0,1)\times (-1/2,1/2]$,
  \begin{equation}
    \label{eq:3}
    \limsup_{N\to+\infty}\left(g(\ln N)\,
      \frac{|S(N,a,b)\,|}{\sqrt{N}}\right)<\infty
    \quad\Longleftrightarrow \quad \sum_{l\geq 1}g^6(l)<\infty.
  \end{equation}
\end{Th}
\noindent This result should be compared with the following theorem
for the exponential sum $S(N,a,b)$ for $b$ in the set
\begin{equation*}
  B_a=\left\{\{\frac12(ma+n)\}_0;\ (m,n)\in\Z^2\setminus(2\Z+1)^2\right\} 
\end{equation*}
where, for $x\in\R$, $\{x\}_0=x\,\mod 1$ and $-1/2<\{x\}_0\leq 1/2$.
For every irrational $a$, the set $B_a$ is dense in $(-1/2,1/2]$ as
the set $\{ma+n;\ (m,n)\in\Z^2\}$
is dense in $\R$.\\
One has
\begin{Th}
  \label{thr:2a}
  Let $g:\R_+\to\R_+$ be a non increasing function. Then, for almost
  all $a\in(0,1)$, there exists a dense $G_\delta$, say $\tilde B_a$,
  such that $B_a\subset\tilde B_a$ and, for $b\in\tilde B_a$, one has
  \begin{equation*}
    \limsup_{N\to+\infty}\left(g(\ln N)\,
      \frac{|S(N,a,b)\,|}{\sqrt{N}}\right)<\infty
    \quad\Longleftrightarrow \quad \sum_{l\geq 1}g^4(l)<\infty.
  \end{equation*}
\end{Th}
\noindent For $b=0\in B_a$, Theorem~\ref{thr:2a} was proved
in~\cite{MR0563894}.\\
Let $\varphi(N)=(\ln N)^{1/4}$. Theorems~\ref{thr:2} and~\ref{thr:2a}
show that for a typical $a$, whereas for $b\in\tilde B_a$ the ratio
$S(N,a,b)/\sqrt{N}$ grows faster than $\varphi(N)$, for a typical $b$,
the ratio $S(N,a,b)/\sqrt{N}$ grows slower than
$(\,\varphi(N)\,)^{2/3+\varepsilon}$ for any $\varepsilon>0$.\\[2mm]
%
%
\begin{floatingfigure}{4cm}
  \includegraphics[bbllx=71,bblly=572,bburx=240,bbury=721,width=3.5cm]{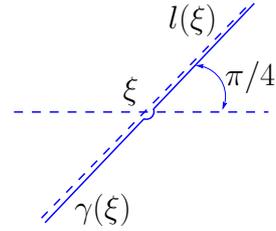}
  \caption{The path}\label{paths}
\end{floatingfigure}
%
\noindent The paper is organized as follows. In section~\ref{sec:Fc}, we
describe the special function mentioned above. Then,
section~\ref{sec:ren-for} is devoted to the exact renormalization
formula, its proof and some useful consequences. It is then used in
section~\ref{sec:asymptotics-exp-sum} to compute asymptotics for
$S(N,a,b)$ when an element of the continuous fraction defining $a$ is
large. Section~\ref{sec:analysis-curlicues} is devoted to the
discussion of the graphs of the quadratic sums and the appearance of
the Cornu spiral. In section~\ref{sec:estimates-S}, we compute precise
estimates of $S(N,a,b)$ in terms of the trajectory of a dynamical
system related to the continued fractions expansion of $a$. Finally,
sections~\ref{sec:almost-sure-growth}
and~\ref{sec:proof-theor-refthr:2} are devoted to the proofs of
Theorems~\ref{thr:2} and~\ref{thr:2a}. The proofs are based on the
estimates obtained in the previous section and on the analysis of
certain dynamical systems.
\section{The special function $\Fc$}
\label{sec:Fc}
Consider the function $\Fc:\C\to \C$ defined by
\begin{equation}
  \label{eq:F}
  \hskip-3cm  \Fc(\xi,a)=\int_{\gamma(\xi)}
  \frac{e\left(\frac{p^2}{2a}\right)\,dp}{e(p-\xi)-1}
\end{equation}
where the contour $\gamma(\xi)$ is going up from infinity along
$l(\xi)$, the strait line $\xi+e^{i\pi/4}\R$, coming infinitesimally
close to the point $\xi$, then, going around this point in the
anti-clockwise direction along an infinitesimally small semi-circle,
and, then, going up to infinity again along $l(\xi)$ (see
Fig.~\ref{paths}).\\
The function $\Fc$ is the special function mentioned in the
introduction. We prove:
\begin{Le}
  \label{Le:diff-eq}
  For each $a>0$, \ $\Fc$ is an entire function of $\xi$, and, for all
  \\ $\xi\in\C$, one has
  \begin{gather}
    \label{eq:F-eq}
    \Fc(\xi,a)-\Fc(\xi-1,a)=e\left(\frac{\xi^2}{2a}\right),\\
    \label{eq:G-eq}
    \Gc(\xi+a,a)-\Gc(\xi,a)=e\left(-\frac{\xi^2}{2a}\right),
  \end{gather}
  where
  \begin{equation}
    \label{eq:G-F}
    \Gc(\xi,a)=c(a)\, e\left(-\frac{\xi^2}{2a}\right)\,\Fc(\xi,a),
    \quad c(a)=e(-1/8)\,a^{-1/2},
  \end{equation}
  and
  \begin{equation}
    \label{eq:f-sym}
    \Fc(-\xi,a)+\Fc(\xi,a)=e\left(\frac{\xi^2}{2a}\right)-\frac1{c(a)}.
  \end{equation}
\end{Le}
\begin{proof}
  The first relation~\eqref{eq:F-eq} follows from the residue
  theorem. The second relation~\eqref{eq:G-eq} becomes obvious after
  the change of variable $z=p-\xi$ in the integral defining $\Fc$.  To
  get the last relation, in the integral representing $\Fc(-\xi,a)$,
  we change the variable $p\to -p$, and then, using the residue
  theorem, we get
  \begin{equation*}
    \Fc(-\xi,a)=e\left(\frac{\xi^2}{2a}\right)-
    \int_{\gamma(\xi)} \frac{e\left(\frac{p^2}{2a}\right)\,e(p-\xi)\,dp}
    {e(p-\xi)-1}.
  \end{equation*}
  This and~(\ref{eq:F}) implies~(\ref{eq:f-sym}). This completes the
  proof of Lemma~\ref{Le:diff-eq}.
\end{proof}
\noindent Lemma~\ref{Le:diff-eq} shows that the function $\Fc$
simultaneously satisfies two difference equations,~(\ref{eq:F-eq})
and~(\ref{eq:G-eq}), with two different shift parameters, $1$ and
$a$. This leads to the renormalization described in the next
section.\\
For small $a$, the asymptotics of $\Fc$ are described by
\begin{Pro}
  \label{pro:Fc-as}
  Let $-1/2\leq \xi\leq 1/2$ and $0<a<1$. Then, $\Fc$ admits the
  representation:
  \begin{equation}
    \label{eq:Fc-as}
    \begin{split}
      \Fc(\xi,a)&=e(1/8)\,f(a^{-1/2}\xi)+O(a^{1/2}),\\
      f(t)&:=e(t^2/2)F(t)\text{ and }F(t):=\int_{-\infty}^t
      e(-\tau^2/2)\,d\tau,
    \end{split}
  \end{equation}
  where $O(a^{1/2})$ is bounded by $C\,a^{1/2}$, and $C$ is a constant
  independent of $a$ and $\xi$.
\end{Pro}
\noindent This is Proposition 1.1 in~\cite{MR2182060}; for the
readers convenience, we repeat its short proof below.\\
For small values of $a$, the special function $\Fc$ ``becomes'' the
Fresnel integral. This proposition and our renormalization formula
immediately explain the curlicues seen in the graphs of the
exponential sums (see e.g.~\cite{MR1140726,MR928946}). Details can be
found in section~\ref{sec:analysis-curlicues} and in~\cite{MR2182060}.
\begin{proof}[Proof of Proposition~\ref{pro:Fc-as}]
  We represent $\Fc$ in the form:
  \begin{equation}
    \label{eq:Fc-as:1}
    \Fc(\xi,a)= \frac1{2\pi i}\int_{\gamma(\xi)}
    \frac{e\left(\frac{p^2}{2a}\right)\,dp} 
    {p-\xi}+\int_{\gamma(\xi)} g(p-\xi)
    \,e\left(\frac{p^2}{2a}\right)\,dp,
  \end{equation}
  where
  \begin{equation*}
    g(p-\xi)=\frac1{e(p-\xi)-1}-\frac1{2\pi
      i(p-\xi)}.
  \end{equation*}
  As $-1/2\leq \xi\leq 1/2$, the integration contour in the second
  integral can be deformed into the curve $\gamma(0)$ without
  intersecting any pole of the integrand. Then, the distance between
  the integration contour and these poles becomes bounded from below
  by $1/2^{3/2}$, and one easily gets $|g(p-\xi)|\leq C$ uniformly in
  $-1/2\le\xi\le1/2$ and in $p\in\gamma(0)$. This immediately implies
  that the second term in~\eqref{eq:Fc-as:1} is bounded by
  $C\,a^{1/2}$. Finally, it is easily seen that first term satisfies
  the equation $I'(\xi)=e(1/8)a^{-1/2}+2i\pi\xi a^{-1} I(\xi)$, and
  that it tends to $0$ when $\xi\to -\infty$ along $\R$.  This implies
  that this term is equal to $e(1/8)\,f(a^{-1/2}\xi)$ and completes
  the proof of Proposition~\ref{pro:Fc-as}.
\end{proof}
\section{Exact renormalization formulas}
\label{sec:ren-for}
We now present exact renormalization formulas for the quadratic
exponential sum $S(N,a,b)$ in terms of the special function
$\Fc(\xi,a)$.
\subsection{One renormalization}
\label{sec:one-renormalization}
One has
\begin{Th}
  \label{th:S2}
  Fix $N\in\N$ and $(a,b)\in(0,1)\times(-1/2,1/2]$. Let
  \begin{gather}
    \nonumber
    \xi=\{a N\},\quad N_1=[a N],\\
    \label{eq:ab}
    a_1=\left\{\frac1a\right\},\quad\quad b_1\equiv \left\{-\frac
      ba+\frac12\,\left[\frac1a\right]\right\}_0,
  \end{gather}
  where $\{x\}$ and $\left[x\right]$ denote the fractional and the
  integer parts of the real number $x$, and $\{x\}_0=x\,\mod 1$ and
  $-1/2<\{x\}_0\leq 1/2$. Then,
  \begin{equation}
    \label{eq:exa-ren-for}
    \begin{split}
      S(N,a,b)&=c(a)\left[e\left(\frac{b^2}{2a}\right)
        \overline{S(N_1,a_1,b_1)}\right.
      \\&\hskip1cm\left.+e\left(-\frac{a N^2}{2}+Nb\right)
        \,\Fc(\xi-b,a)-\Fc(-b,a)\right].
    \end{split}
  \end{equation}
\end{Th}
\noindent To our knowledge, such renormalization formulas (though
without explicit description of the terms containing $\mathcal F$)
first appeared in~\cite{MR1555099} and have since then a long
tradition. The formula~\eqref{eq:exa-ren-for} is analogous to the less
general one derived in~\cite{MR2182060}. It should also be compared to
Theorems 4 and 5 in~\cite{MR0563894}.
\begin{proof}[Proof of Theorem~\ref{th:S2}]
  The idea of the proof is to compute the quantity ${\mathcal
    F}(Na-b,a)$ in two different ways: first,
  using~(\ref{eq:G-eq}), and then, using (\ref{eq:F-eq}).\\
  By means of~(\ref{eq:G-eq}), we get
  \begin{align*}
    \Gc(Na-b,b)&=\sum_{k=0}^{N-1}e\,\left(-\frac{(ka-b)^2}
      {2a}\right) +\Gc(-b,a)\\
    &=e\left(-\frac{b^2}{2a}\right)\,S(N,a,b)+\Gc(-b,a).
  \end{align*}
  Note that this relation and~(\ref{eq:G-F}) imply that
  \begin{equation}
    \label{eq:s-g}
    S(N,a,b)=c(a)\,
    \left[ e\,\left(-\frac{N^2a}2+Nb\right)\,\Fc(Na-b,a)-
      \Fc(-b,a)\right]
  \end{equation}
  On the other hand, using~(\ref{eq:F-eq}), we obtain
  \begin{align*}
    \Fc(Na-b,b)&-\Fc(\xi-b,a)=\sum_{k=0}^{N_1-1}e\,
    \left(\frac{(Na-k-b)^2}{2a}\right)\\
    &=e\,\left(\frac{(Na-b)^2}{2a}\right)\,
    \sum_{k=0}^{N_1-1}e\,\left(\frac{k^2}{2a}-\frac{k(Na-b)}a\right).
  \end{align*}
  As $e(l)=1$ for all $l\in\Z$, and as, modulo 1, one has
  \begin{equation*}
    \begin{split}
      \frac{k^2}{2a}+\frac ba k&= \frac{k(k+1)}{2}\frac1a+ \left(\frac
        ba-\frac1{2a}\right) k \\&=
      \frac{k(k+1)}{2}a_1-k\left(b_1+\frac{a_1}2\right)=
      \frac{k^2}{2}a_1-kb_1,
    \end{split}
  \end{equation*}
  we get finally
  \begin{equation*}
    \Fc(Na-b,b)=
    e\,\left(\frac{(Na-b)^2}{2a}\right)\,\overline{S(N_1,a_1,b_1)}
    +\Fc(\xi-b,a).  
  \end{equation*}
  Plugging this formula into~(\ref{eq:s-g}), we
  obtain~(\ref{eq:exa-ren-for}). This completes the proof of
  Theorem~\ref{th:S2}.
\end{proof}
\subsection{Multiple renormalizations}
\label{sec:mult-renorm}
The renormalization formula~(\ref{eq:exa-ren-for}) expresses the
Gaussian sum $S(N,a,b)$ in terms of the sum $S(N_1,a_1,b_1)$
containing a smaller number of terms. We can renormalize this new sum
and so on. After a finite number of renormalizations, the number of
terms in the exponential sum is reduced to one. Let us now describe
the formulas obtained in this way when $a$ is irrational.\\
For $l\geq 0$, we let
\begin{gather}
  \label{eq:als}
  a_{l+1}= \left\{\frac1{a_{l}}\right\},\quad a_0=a,\quad\quad
  N_{l+1}=[a_lN_l],\quad N_0=N,\\
  \label{eq:bls}
  b_{l+1}\equiv
  \left\{-\frac{b_l}{a_l}+\frac12\left[\frac1{a_l}\right]\right\}_0\,,
  \quad b_0=b.
\end{gather}
In the sequel, when required, we will sometimes write $N_l(N)=N_l$,
$b_l(b)=b_l$ and $a_l(a)=a_l$ to mark the dependency
on the initial value of the sequence.\\
The sequence $\{N_l\}$ is strictly decreasing until it reaches the
value zero and then becomes constant. Denote by $L(N)$ the unique
natural number such that
\begin{equation}
  \label{eq:LdeN}
  N_{L(N)+1}=0\quad\text{ and }\quad N_{L(N)}\geq 1.
\end{equation} 
Theorem~\ref{th:S2} immediately implies
\begin{Cor} \label{cor:mult-ren} One has
  \begin{equation}
    \label{eq:S-multiple}
    S(N,a,b)=
    \sum_{l=0}^{L}\frac{e(\theta_l)}{(a_0a_1\dots
      a_{l})^{1/2}}\,
    \Delta \Fc_l^{*l},\\
  \end{equation}
  where
  \begin{equation}
    \label{eq:delta-F}
    \Delta\Fc_l=e(-a_lN_l^2/2+N_lb_l)\Fc(\xi_l-b_l,a_l)-
    \Fc(-b_l,a_l).
  \end{equation}
  and
  \begin{itemize}
  \item $*l$ denotes the complex conjugation applied $l$ times
  \item $\xi_{l}=\{a_{l}N_{l}\}$,
  \item $\theta_{l+1}=\theta_l+(-1)^l\left(\frac18+\frac{b_l^2}
      {2a_l}\right)$ where $\theta_0=-1/8$.
  \end{itemize}
\end{Cor}
\section{Asymptotics of the exponential sum}
\label{sec:asymptotics-exp-sum}
From formula~(\ref{eq:S-multiple}), we now derive a representation for
$S(N,a,b)$ that, for small values of $a_L$, becomes an asymptotic
representation. This representation explains the curlicues structures
in the graphs of the exponential sum that we have mentioned already
and that are shown in Figure~\ref{fig:fresnel}.
\subsection{Preliminaries}
\label{sec:preliminaries}
We first discuss some analytic objects used to describe the
asymptotics of the exponential sums.\\
Recall that $L(N)$ is defined by~(\ref{eq:LdeN}).  The function $N\to
L(N)$ is a non-decreasing function of $N$. Define
\begin{equation}
  \label{eq:1}  N^-(L)=\min\{N;\ L(N)=L\}\text{ and } N^+(L)=\max\{N;\
  L(N)=L\}
\end{equation}
Clearly,
\begin{equation}\label{Npm}
  N^+(L-1)=N^-(L)-1. 
\end{equation}
One has
\begin{Le}
  \label{le:1}
  Let $L\in\N$. Then
  \begin{equation}
    \frac1{a_0a_1\dots a_{L-1}}<N^-(L)< 
    \frac1{a_0a_1\dots a_{L-1}}\,(1+4a_{L-1}).
  \end{equation}
\end{Le}
\begin{proof}[Proof of Lemma~\ref{le:1}]
  Using the definition of $N^\pm(L)$, we get
  \begin{gather}
    \label{eq:nm-low}
    1\leq [a_{L-1}[\dots[a_1[a_0N^-(L)]]\dots]<
    a_{L-1}\dots a_1a_0N^-(L);\\
    \label{eq:nm-up}
    \begin{split}
      1&> a_L [a_{L-1}[\dots[a_1[a_0N^+(L)]]\dots] \\&> a_{L}\dots
      a_1a_0N^+(L)-a_L-a_La_{L-1}\dots - a_L\dots a_2a_1\\
      &=a_{L}\dots a_1a_0N^-(L+1)-a_L-a_La_{L-1}\dots - a_L\dots
      a_1a_0.
    \end{split}
  \end{gather}
  Inequality~(\ref{eq:nm-low}) implies the lower bound for
  $N^-(L)$.\\
  To get the upper bound we use the well known estimate
  \begin{equation}
    \label{eq:two-a-l}
    \forall l\in\N\quad a_la_{l-1}<\frac12
  \end{equation}
  that immediately follows from the representation
  $a_{l-1}=\frac1{n_l+a_l}$, where $n_l$ is a positive integer;
  indeed, one computes
  \begin{equation*}
    a_la_{l-1}=1-n_la_{l-1}=1-\frac{n_l}{n_l+a_l}=
    \frac{a_l}{n_l+a_l}<
    \frac1{n_l+1}\leq \frac12.
  \end{equation*}
  Estimates~(\ref{eq:nm-up}) and~(\ref{eq:two-a-l}) imply that
  $a_{L}\dots a_1a_0\,N^-(L+1)<1+4a_{L}$.  This completes the proof of
  Lemma~\ref{le:1}.
\end{proof}
\noindent Fix $L\in\N$ and for $N^-(L)\leq N\leq N^+(L)$, consider the
quantity
\begin{equation}
  \label{eq:last-xi}
  \xi=\xi_L(N)=a_{L}N_{L}(N)
\end{equation}
The definitions of $N^-(L)$ and $N^+(L)$, see~(\ref{eq:1}), imply
\begin{equation}
  \label{eq:xi-bounds}
  a_L\leq \xi_L(N)<1.
\end{equation}
On the interval $N^-(L)\leq N\leq N^+(L)$, the function $N\to
\xi_L(N)=\xi$ is a non-decreasing function of $N$. One has
\begin{Le}\label{le:xi-possible-values}
  As $N$ increases from $N^-(L)$ to $N^+(L)$, $\xi_L(N)$ runs through
  all the values $a_L$, $2a_L$, $3a_L$, $\dots $ that are smaller than
  1.
\end{Le}
\begin{proof}[Proof of Lemma~\ref{le:xi-possible-values}]
  For $l\in\N$, define $\tilde \xi_{l+1}(N)=a_l[\tilde \xi_l(N)]$
  where $\tilde\xi_0(N)=a_0 N$. All the functions $N\mapsto\tilde
  \xi_l(N)$ are non-decreasing functions of $N$ such that
  \begin{itemize}
  \item $\tilde\xi_l(0)=0$ and $\tilde\xi_l(N)\to+\infty$ as
    $N\to+\infty$,
  \item $\tilde\xi_l(N)=0$ if $N< N_-(L)$ and $\tilde\xi_l(N)>1$ if
    $N> N_+(L)$,
  \item $\tilde \xi_{L(N)}(N)=\xi_{L(N)}(N)$ where $L(N)$ is defined
    in~(\ref{eq:LdeN}).
  \end{itemize}
  So, it suffices to check that, for fixed $l$, $\tilde\xi_l(N)$ takes
  all the values $a_l$, $2a_l$, $3a_l$, $\dots$ as $N$ increases. For
  $l=0$, this is obvious. Assume that it holds for some $l>0$.  Show
  that, for any $m\in\N^*$, $\tilde\xi_{l+1}(N)$ takes the value
  $ma_{l+1}$ for some $N$. Pick $m\in\N^*$ and consider the largest
  $N$ such that $\tilde\xi_l(N)<m$. One has $m\leq\tilde\xi_l(N+1)=
  \tilde\xi_l(N)+a_l<m+1$. So, $\tilde\xi_{l+1}(N+1)=
  m\,a_{l+1}$. This completes the proof of
  Lemma~\ref{le:xi-possible-values}.
\end{proof}
\subsection{Asymptotics}
\label{sec:asymptotics}
Recall that $-1/2<b_L\leq 1/2$. We prove
\begin{Th}
  \label{th:as-exp-sum}
  Let $a$ be irrational and $L$ be a positive integer. Assume that $N^-(L)\leq
  N \le N^+(L)$. Define $\xi_L(N)$ by~\eqref{eq:last-xi}.  Then, for
  $\xi_L(N)-b_L\leq 1/2$, one has
  \begin{equation}\label{eq:xi-le-half}
    S(N,a,b)=\frac{e(\theta_{L+1})}{\sqrt{a_0a_1\dots a_L}}\,
    \left(\int_{-\frac{b_L}{\sqrt{a_L}}}^{\frac{\xi_L(N)-b_L}{\sqrt{a_L}}}
      e(-\tau^2/2)\,d\tau+O\left(\sqrt{a_L}\right)\right)^{*L} 
  \end{equation}
  and, for $\xi_L(N)-b_L\geq 1/2$, one has
  \begin{equation}
    \label{eq:xi-ge-half}
    \begin{split}
      S(N,a,b)&= \frac{e(\theta_{L+1})}{\sqrt{a_0a_1\dots a_L}}\,
      \left(\int_{-\frac{b_L}{\sqrt{a_L}}}^{\infty}e(-\tau^2/2)\,d\tau
        +O\left(\sqrt{a_L}\right)+\right.  \\ &\hskip1.2cm+{\dsize
        e\left(\frac{b_L-\xi_L(N)+1/2}{2a_L}\right)} \left.
        \int_{\frac{1-(\xi_L(N)-b_L)}{\sqrt{a_L}}}^{\infty}
        e(-\tau^2/2)\,d\tau\right)^{*L}
    \end{split}
  \end{equation}
  where $*L$ and $\theta_{l}$ are defined in
  Corollary~\ref{cor:mult-ren}.
\end{Th}
\noindent Formulas~\eqref{eq:xi-le-half} and~\eqref{eq:xi-ge-half}
give asymptotics for $S(N,a,b)$ when $a_L$ is small.
\begin{proof}[Proof of Theorem~\ref{th:as-exp-sum}]
  As we will see later on, the $L$-th term in~(\ref{eq:S-multiple}) is
  the leading term in this expansion. To get the formulas for the
  leading term, let us study the expression for $\Delta\Fc_L$. To
  simplify the notations, we write $\xi=\xi_L(N)$.
  By~(\ref{eq:delta-F}) and~(\ref{eq:last-xi}),
  \begin{equation}
    \label{eq:D-F-1}
    \Delta\Fc_L=
    e\left(-\frac{\xi^2}{2a_L}+\frac{\xi b_L}{a_L}\right)\,
    \Fc(\xi-b_L,a_L)-\Fc(-b_L,a_L).
  \end{equation}
  Now, assume that $\xi-b_L\leq 1/2$. Replacing $\Fc$ by its
  representation~(\ref{eq:Fc-as}), we get
  \begin{equation}
    \label{eq:D-F-2}
    \Delta\Fc_L=
    e\left(\frac{b_L^2}{2a_L}+\frac{1}{8}\right)\,
    \int_{-\frac{b_L}{\sqrt{a_L}}}^{\frac{\xi-b_L}{\sqrt{a_L}}}
    e(-\tau^2/2)\,d\tau+O(\sqrt{a_L}).
  \end{equation}
  this implies that, up to the term
  $\frac{O(\,\sqrt{a_L}\,)}{\sqrt{a_0a_1\dots a_L}}$, the $L$-th term
  in~(\ref{eq:S-multiple}) coincides
  with the leading term in~(\ref{eq:xi-le-half}).\\
  Assume that $\xi-b_L\geq 1/2$. Now, we express
  $\Delta\Fc(\xi-b_L,a_L)$ in terms of $\Delta\Fc( (1-(\xi-b_L),a_L)$
  that can be directly described
  by~(\ref{eq:Fc-as}). By~(\ref{eq:f-sym}) and~(\ref{eq:F-eq}), we get
  \begin{equation}
    \label{eq:Fc-sym1}
    \Fc(\xi,a)=-\Fc(1-\xi,a)+e\left(\frac{\xi^2}{2a}\right)+
    e\left(\frac{(1-\xi)^2}{2a}\right)-1/c(a).
  \end{equation}
  This and~\eqref{eq:D-F-1} imply that
  \begin{equation}
    \label{eq:D-F-3}
    \begin{split}
      \Delta\Fc_L&= e\left(-\frac{\xi^2}{2a_L}+\frac{\xi
          b_L}{a_L}\right)\, \left[
        e\left(\frac{(\xi-b_L)^2}{2a_L}\right)+
        e\left(\frac{(1-\xi+b_L)^2}{2a_L}\right)\right.\\&\hskip2.5cm\left.
        -\frac1{c(a_L)}-\Fc(1-\xi+b_L,a_L)\right]-\Fc(-b_L,a_L).
    \end{split}
  \end{equation}
  As $0<\xi<1$, $|b_L|\leq 1/2$ and $\xi-b_L\geq 1/2$, one has
  $-1/2\leq 1-(\xi-b_L)< 1/2$. So, in~(\ref{eq:D-F-3}), we replace
  $\Fc$ by its representation~(\ref{eq:Fc-as}) and use
  \begin{equation*}
    \int_{-\infty}^{\infty}e(-\tau^2/2)\,d\tau=e(-1/8)\text{ and }
    1/c(a)=O(\sqrt{a}),
  \end{equation*}
  to get
  \begin{equation}
    \label{eq:D-F-4}
    \begin{split}
      \Delta\Fc_L&= {\dsize
        e\left(\frac{b_L^2}{2a_L}+\frac{1}{8}\right)}\,\, \left(
        \int_{-\frac{b_L}{\sqrt{a_l}}}^{\infty}e(-\tau^2/2)\,d\tau\right.
      \\&\hskip1cm\left. + {\dsize
          e\left(\frac{b_L-\xi+1/2}{a_L}\right)}
        \int_{\frac{1-(\xi-b_L)}{\sqrt{a_L}}}^{\infty}
        e(-\tau^2/2)\,d\tau+O\left(\sqrt{a_L}\right)\,\,\right).
    \end{split}
  \end{equation}  
  When $\xi-b_L\geq 1/2$, this implies that, the $L$-th term
  in~(\ref{eq:S-multiple}) coincides with the leading term
  in~(\ref{eq:xi-ge-half}) up to
  $O(\frac{\sqrt{a_L}}{\sqrt{a_0a_1\dots
      a_L}})$.\\
  To complete the proof, we have to estimate the contribution to
  $S(N,a,b)$ of the sum $\sum_{l=0}^{L-1}\dots$
  in~(\ref{eq:S-multiple}). It follows from
  Proposition~\ref{pro:Fc-as} and equation~(\ref{eq:F-eq})
  $\xi\mapsto\Fc(\xi,a)$ is locally bounded, uniformly in $a$. This
  observation and~(\ref{eq:two-a-l}) imply the uniform estimate
  \begin{equation}
    \label{eq:error-S}
    \left|\sum_{l=0}^{L-1}\frac{e(\theta_l)}
      {(a_0a_1\dots a_{l})^{1/2}}\,
      \Delta \Fc_l^{*l}\right|\leq \frac{C}
    {(a_0a_1\dots a_{L-1})^{1/2}}.
  \end{equation}
  This estimate,~(\ref{eq:D-F-2}) and~(\ref{eq:D-F-4})
  imply~(\ref{eq:xi-le-half}) and~(\ref{eq:xi-ge-half}). This
  completes the proof of Theorem~\ref{th:as-exp-sum}.
\end{proof}
\noindent The following corollary of Theorem~\ref{th:as-exp-sum} will
be of use later on.
\begin{Cor}
  \label{cor:S-est}
  Fix $L\in\N$ and $N^-(L)\leq N\leq N^+(L)$. Write
  $\xi=\xi_L(N)$. For $\xi-b_L\leq 1/2$,
  \begin{equation}
    \label{ineq:xi-le-half}
    \begin{split}
      \left|\frac{ S(N,a,b)}{\sqrt{N}}\right|&=
      \left|\frac{1+O(a_L/\xi)}{\sqrt{\xi}}\,
        \int_{-\frac{b_L}{\sqrt{a_L}}}^{\frac{\xi-b_L}{\sqrt{a_L}}}
        e(-\tau^2/2)\, d\tau\right|\\&\hskip6cm
      +O\left(\sqrt{a_L/\xi}\,\right),
    \end{split}
  \end{equation}
  and, for $\xi-b_L\geq 1/2$
  \begin{equation}
    \label{ineq:xi-ge-half}
    \begin{split}
      \left|\frac{S(N,a,b)}{\sqrt{N}}\right|&\leq
      \frac{C}{\sqrt{\xi}}\,
      \left(\,\left|\int_{-\frac{b_L}{\sqrt{a_L}}}
          ^{\infty}e(-\tau^2/2)\,d\tau\right|\right.  \\&\hskip1.5cm
      \left. +\left|\int_{\frac{1-(\xi-b_L)}{\sqrt{a_L}}}^{\infty}
          e(-\tau^2/2)\,d\tau\right|\right)
      +O\left(\sqrt{a_L/\xi}\,\right).
    \end{split}
  \end{equation}
  The error terms estimates are uniform in $L$, $N$, $a$ and $b$.
\end{Cor}
\begin{proof}
  The corollary follows from Theorem~\ref{th:as-exp-sum}, the
  representation
  \begin{equation}
    \label{eq:xi-aN}
    \frac1{\sqrt{a_L\dots a_1a_0 N}}=
    \frac{1+O(a_L/\xi)}{\sqrt{\xi}},\quad\quad 
    |O(a_L/\xi)|\leq 4 a_L/\xi,
  \end{equation}
  and the lower bound from~(\ref{eq:xi-bounds}).  To
  check~(\ref{eq:xi-aN}), we note that~(\ref{eq:last-xi}) implies that
  \begin{equation*}
    \begin{split}
      a_L\dots a_1a_0N\geq \xi&\geq a_L\dots a_1a_0N-
      a_L-a_La_{L-1}\dots-a_L\dots a_1\\&> a_L\dots a_1a_0N-4a_L,
    \end{split}
  \end{equation*}
  and as $\xi\geq a_L$, these estimates imply~(\ref{eq:xi-aN}). This
  completes the proof of Corollary~\ref{cor:S-est}.
\end{proof}
\subsection{Analysis of the curlicues}
\label{sec:analysis-curlicues}
The formulas~\eqref{eq:xi-le-half} and~\eqref{eq:xi-ge-half} and
Lem\-ma~\ref{le:xi-possible-values} explain the curlicue structures
seen in the graphs of the exponential sums and discussed in many
papers (see e.g.~\cite{MR1140726,MR928946,MR1443869}).\\
The graph of an exponential sum is just the graph obtained by linearly
interpolating between the values of $S(N,a,b)$ obtained for
consecutive $N$. In Fig.~\ref{fig:fresnel}, we show an example of such
a graph.
%
\begin{figure}[htbp]
  \centering \subfigure[The graph of a sum]{
    \includegraphics[bbllx=14,bblly=19,bburx=581,bbury=822,angle=270,width=5cm]{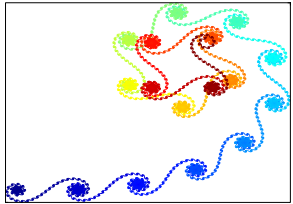}
    \label{fig:graphe}} \hskip1cm \subfigure[A zoom of a detail of
  this graph]{
    \includegraphics[bbllx=14,bblly=19,bburx=581,bbury=822,angle=270,width=5cm]{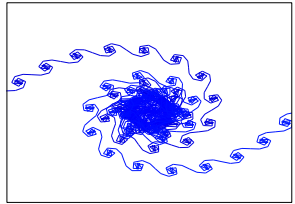}
    \label{fig:zoom_graph}}
  \caption{The graph of an exponential sum}
  \label{fig:fresnel}
\end{figure}
%
One distinctly sees the spiraling structure that were dubbed curlicues
in~\cite{MR928946}. These are seen for $N$ such that $a_{L(N)}$ is
small; indeed, in this case, as formulas~\eqref{eq:xi-le-half}
and~\eqref{eq:xi-ge-half} show, up to a rescaling and possibly a
shift, the graph of the exponential sum is obtained by sampling points
on the graph of the Fresnel integral, the Cornu spiral.  Thanks to
formulas~\eqref{eq:xi-le-half} and~\eqref{eq:xi-ge-half}, one can
compute all the geometric characteristics of the curlicues when
$a_{L(N)}$ is small.\\
In Fig.~\ref{fig:zoom_graph}, we zoomed in on one of the curlicues
shown in Fig.~\ref{fig:graphe}. Now we see the curlicues from the
``previous generation''. They are seen in the case where $a_{L-1}$ is
small and can be explained by the asymptotic analysis of the $(L-1)$th
term in~\eqref{eq:S-multiple}.
\section{Estimates on the exponential sums}
\label{sec:estimates-S}
Using Theorem~\ref{th:as-exp-sum}, we now estimate $S(N,a,b)$
in terms of the sequences $(a_l)_l$ and $(b_l)_l$.\\
For $L\in\N$, define
\begin{equation}
  \label{eq:M}
  M(L,a,b)=\max_{N^-(L)\leq N\leq N^+(L)}
  \left|\frac{S(N,a,b)}{\sqrt{N}}\right|
\end{equation}
We prove
\begin{Pro}
  \label{pro:est-s}
  There exist positive constants $c$ and $C$ independent of $a$, and
  $b$ such that, for $L\in\N$,
  \begin{equation}
    \label{eq:est-s:up}
    M(L,a,b)\leq C\,\frac1{\sqrt{|b_L|}+\sqrt[4]{a_L}},\\
  \end{equation}
  and
  \begin{equation}
    \label{eq:est-s:down}
    \text{if \ } \sqrt{|b_l|}+\sqrt[4]{a_L}\leq c, \text{ \ then \ }
    \frac1C\,\frac1{\sqrt{|b_L|}+\sqrt[4]{a_L}} \leq 
    M(L,a,b)
  \end{equation}
  where $a_L$ and $b_L$ are defined by~(\ref{eq:als})
  and~(\ref{eq:bls}).
\end{Pro}
\begin{proof}[Proof of Proposition~\ref{pro:est-s}]
  {\it Preliminaries.} \ Below, we consider only $N$ satisfying
  $N^-(L)\leq N\leq N^+(L)$.  All the constants $C$ in the proof
  are independent of $L$, $N$, $a$ and $b$. \\
  The analysis is based on Corollary~\ref{cor:S-est}.  To
  obtain~(\ref{eq:est-s:up}) from Corollary~\ref{cor:S-est}, we
  systematically use three simple estimates
  \begin{gather}
    \label{est:1}
    \forall x,y\in\R,\quad
    \left| \int_{x}^{y}e(-\tau^2/2)\,d\tau\right|\leq C,\\
    \label{est:2}
    \forall x,y\in\R,\quad
    \left| \int_{x}^{y}e(-\tau^2/2)\,d\tau\right|\leq |x-y|,\\
    \label{est:3}
    \forall x>0,\quad \left| \int_{\pm\infty}^{\pm
        x}e(-\tau^2/2)\,d\tau\right|\leq \frac{C}{x},
  \end{gather}
  To simplify the notations, we write $\xi=\xi_L(N)$. Recall that
  $|b_L|\leq 1/2$ and $a_L\leq \xi\leq 1$.  Note that this implies
  that $-1/2\leq -b_L+\xi\leq 3/2$. First, we get upper bounds for
  $S(N,a,b)/\sqrt{N}$. Therefore, depending on the values of $\xi$ and
  $b_L$, we consider several cases.
  \begin{itemize}
  \item Let $-b_L+\xi\geq 1/2$. One has
    \begin{equation}
      \label{eq:case1}
      \left|\frac{S(N,a,b)}{\sqrt{N}}\right|\leq C.
    \end{equation}
    If $\xi\geq 1/4$, this estimate follows
    from~(\ref{ineq:xi-ge-half}) and~(\ref{est:1}). If $\xi\leq 1/4$
    then $-b_L\geq 1/4$ and $1-(\xi-b_L)\geq 1/4$. We estimate both
    integrals in~(\ref{ineq:xi-ge-half}) using~(\ref{est:3}) to obtain
    $ \left|\frac{S(N,a,b)}{\sqrt{N}}\right|\le
    C\sqrt{a_L/\xi}$. Then,~(\ref{eq:xi-bounds})
    yields~(\ref{eq:case1}).
  \item Let $-b_L+\xi\leq 1/2$. We now have to consider three
    sub-cases depending on the value of $b_L$. In all these cases, we
    base our analysis on~(\ref{ineq:xi-le-half}).
    By~(\ref{eq:xi-bounds}) the terms $1+O(a_L/\xi)$ and
    $O(\sqrt{a_l/\xi})$ in this formula are bounded by a constant, and
    we only have to estimate the term $T=\left|\frac{1}{\sqrt{\xi}}\,
      \int_{-\frac{b_L}{\sqrt{a_L}}}^{\frac{\xi-b_L}{\sqrt{a_L}}}
      e(-\tau^2/2)\, d\tau\right|$.
    \begin{itemize}
    \item When $-b_L\ge\sqrt{a_L}$, one has
      \begin{equation}
        \label{eq:case2a}
        T\leq \frac{C}{\sqrt{|b_L|}}. 
      \end{equation}
      If $\xi\leq a_L/|b_L|$, one estimate the integral
      using~(\ref{est:2}), otherwise one uses~(\ref{est:3}). In both
      cases, this yields~\eqref{eq:case2a}.
    \item When $ |b_L|\leq \sqrt{a_L}$, one has
      \begin{equation}
        \label{eq:case2b}
        T\leq \frac{C}{\sqrt[4]{a_L}}. 
      \end{equation}
      For $\xi\le \sqrt{a_L}$, one uses~(\ref{est:2}), otherwise one
      uses~(\ref{est:1}). This leads to~\eqref{eq:case2b}.
    \item When $ -b_L\leq -\sqrt{a_L}$, one has
      \begin{equation}
        \label{eq:case2c}
        T\leq \frac{C}{\sqrt{|b_L|}}. 
      \end{equation}
      If $\xi\ge b_L/2$, then~(\ref{est:1}) yields~(\ref{eq:case2c}).
      If $\xi\le a_L/b_L$, then we get~(\ref{eq:case2c})
      using~(\ref{est:2}). Now, assume that $\xi\le b_L/2$, and that
      $\xi\ge a_L/b_L$. The first inequality then implies that
      $-b_L+\xi\le-b_L/2$, and, by means of~(\ref{est:3}), we get
      $T\le C\,(\sqrt{a_l/\xi})/ b_L$. As $\xi\ge a_L/b_L$, this
      implies~(\ref{eq:case2c}).
    \end{itemize}
  \end{itemize}
  Estimates~(\ref{eq:case1})~---~(\ref{eq:case2c}) all imply the upper
  bound~(\ref{eq:est-s:up}).
  \vskip.2cm\par To prove the lower bound, we consider the leading
  term in the representations given in Corollary~\ref{cor:S-est} for
  well chosen values of $\xi$. We consider three cases depending on
  the value of $b_L$.
  \begin{itemize}
  \item When $-b_L\leq -\sqrt{a_L}$. Recall that the possible values
    $\xi$ are described in Lemma~\ref{le:xi-possible-values}. Let
    $\xi_0=\frac{a_L}{2b_L}$ and choose $N$ so that $|\xi-\xi_0|\le
    a_L$. Then, one has
    \begin{equation*}
      -b_L+\xi\leq
      -\sqrt{a_L}+\sqrt{a_L}/2+a_L<a_L/2<1/2    ,
    \end{equation*}
    and we can use~(\ref{ineq:xi-le-half}).\\
    Let $t=b_L/\sqrt{a_L}$ and $s=\xi/\xi_0\in[1-2b_L,1+2b_L]$.
    Assuming that $c$ in~(\ref{eq:est-s:down}) is smaller than $1/16$,
    we get $s\in[1/2,3/2]$. \\
    Represent the leading term in~(\ref{ineq:xi-le-half}) in the form
    $(b_L)^{-1/2}g(t,s)$ where
    \begin{equation*}
      g(t,s)=\sqrt{\frac2s}\,t\,\int_{-t}^{-t+\frac{s}{2t}} 
      e(-\tau^2/2)\,d\tau. 
    \end{equation*}
    Note that:
    \begin{enumerate}
    \item $g$ never vanishes as the Cornu spiral i.e. the graph of the
      Fresnel integral $x\to \int_{-\infty}^xe(-\tau^2/2)\,d\tau$,
      $x\in\R$, has no self-intersections,
    \item $|g(t,s)|\to \frac1\pi\sqrt{\frac2s}sin(\pi s/2)$ as
      $t\to\infty$ uniformly in $s$; one checks this by integration by
      parts.
    \end{enumerate}
    Hence, for any $s$, $\D\inf_{t\geq1}|g(t,s)|\geq C>0$. This
    implies that the leading term in~(\ref{ineq:xi-le-half}) is
    bounded away from $0$ by $C/\sqrt{b_L}$.  On the other hand, for
    $\xi=s\xi_0$, the error term in~(\ref{ineq:xi-le-half}) is bounded
    by $C\sqrt{b_L}$.  So, if $\sqrt{b_L}<c$, and $c$ is small enough,
    we see that the right hand side in ~(\ref{ineq:xi-le-half}) is
    bounded away from $0$ by $C/\sqrt{b_L}$.  This completes the proof
    of~(\ref{eq:est-s:down}) in the case where $-b_L\leq -\sqrt{a_L}$.
  \item When $-b_L\geq \sqrt{a_L}$. One proves the lower bound almost
    in the same way as in the previous case. Now, we define
    $\xi_0=\frac{a_L}{2|b_L|}$ and choose $N$ as before. We get
    $-b_L+\xi\le |b_L|+\sqrt{a_L}/2+a_L$, and the last expression is
    smaller than $1/2$ if $c$ in~(\ref{eq:est-s:down})
    is chosen small enough. \\
    We define $s$ as above and let $t=|b_L|/\sqrt{a_L}$. Hence,
    $s\in[1/2,3/2]$ and $t\geq 1$. Then, we write the leading term
    in~(\ref{ineq:xi-le-half}) as $|b_L|^{-1/2}g(t,s)$ where
    \begin{equation*}
      g(t,s)=\sqrt{\frac2s}\,t\,\int_{t}^{t+\frac{s}{2t}}
      e(-\tau^2/2)\,d\tau. 
    \end{equation*}
    The analysis is then analogous to the one done in the previous
    case; we omit further details.
  \item When $|b_L|\leq \sqrt{a_L}$. The plan of the proof remains the
    same as in the previous cases. Now, we define $\xi_0=\sqrt{a_L}$.
    The number $N$ is chosen as before.  We get $-b_L+\xi\le
    2\sqrt{a_L}+a_L$, and so this expression is
    smaller than $1/2$ if $c$ in~(\ref{eq:est-s:down}) is chosen small enough. \\
    We define $s$ as before, and we let $t=b_L/\sqrt{a_L}$. We get
    $|t|\le 1$ and $s\in[1/2,3/2]$ (if $c$ is chosen small enough).
    The leading term in~(\ref{ineq:xi-le-half}) is equal to
    $(a_L)^{-1/4}g(t,s)$, with
    \begin{equation*}
      g(t,s)=\sqrt{\frac1{s}}\,\int_{-t}^{-t+s}e(-\tau^2/2)\,d\tau.
    \end{equation*}
    Again $g\ne0$, and so, on the compact set $(t,s)\in
    [-1,1]\times [1/2,3/2]$, the factor $g$ is bounded away from $0$ by a
    constant $C$.  Now, representation~(\ref{ineq:xi-le-half}) implies
    that
    \begin{equation*}
      |S(N,a,b)|\geq C/\sqrt[4]{a_L}-
      C\sqrt[4]{a_L}
    \end{equation*}
    (if $c$ is chosen small enough), and we
    obtain~(\ref{eq:est-s:down}).
  \end{itemize}
  This completes the proof of the lower bound and, so, the proof of
  Proposition~\ref{pro:est-s}.
\end{proof}
\section{The proof of Theorem~\ref{thr:2}}
\label{sec:almost-sure-growth}
We now turn the proofs of Theorem~\ref{thr:2} and Theorem~\ref{thr:2a}
in the next section. Both will be deduced from
Proposition~\ref{pro:est-s} and the study of certain dynamical
systems.
\subsection{Reduction of the proof of Theorem~\ref{thr:2} to the
  analysis of a dynamical system}
\label{sec:reduct-anlys-dynam}
We first reduce the proof of Theorem~\ref{thr:2} to the proof of two
lemmas describing properties of the dynamical system defined on the
square $K:=[0,1)\times(-1/2,1/2]$ by the formulas~\eqref{eq:als}
and~\eqref{eq:bls}. The idea of such a reduction was inspired to us by
the proof of Theorem II, Chapter 7, from~\cite{MR0306130}.\\
Note that it suffices to prove Theorem~\ref{thr:2} in the case when
\begin{equation}
  \label{eq:g-add-cond}
  \forall l\in\N,\quad|g(l)|\leq 1/2,\quad\text{and}\quad
  \lim_{l\to\infty}g(l)=0
\end{equation}
which we assume from now on.\\
We begin by formulating the two lemmas referred to above.
\smallpagebreak Let $\varphi:\R_+\to \R_+$ be a non increasing
function.  Let $\gamma(a,b)$ be the trajectory of the dynamical system
defined by~\eqref{eq:als} and~\eqref{eq:bls} that begins at $(a,b)\in
K$. Let $\mathfrak N(L,a,b)$ be the number of the conditions
\begin{equation}
  \label{eq:trajectory-conditions}
  \text{`` \ } 
  \sqrt[4]{a_l}\le\varphi(l)\quad\text{and}\quad \sqrt{|b_l|}\leq \varphi(l)
  \text{ \ ''}
\end{equation}
with $0\leq l\leq L$ that are satisfied along $\gamma(a,b)$. Note that
\begin{equation}
  \label{eq:N-chi}
  \mathfrak N(L,a,b)=\mathfrak N(L,\varphi,a,b)=\sum_{l=0}^L \chi(
  \sqrt[4]{a_l}\le\varphi(l))\,\chi(\sqrt{|b_l|}\leq \varphi(l)),
\end{equation}
where $\chi(\text{``statement''})$ is equal to 0 if the ``statement''
is false and is equal to 1 otherwise.\\
Let $m$ be the measure on $K$ defined by the formula $m(D)=\frac1{\ln
  2}\,\int_D\frac{da\,db}{1+a}$ for $D\subset K$ measurable. Note that
$m$ is a probability measure. We denote by $\Vert\mathfrak
N(L,\cdot,\cdot)\Vert_1$ and $\Vert\mathfrak N(L,\cdot,\cdot)\Vert_2$,
respectively, the $L^1(K,\,m)$ and $L^2(K,\,m)$ norms of the function
$(a,b)\to \mathfrak N(L,a,b)$.
\begin{Rem}
  The measure $\frac1{\ln2}\frac{da}{1+a}$ is the invariant measure
  for the Gauss transformation $a\to \left\{\frac1a\right\}$ on
  $(0,1)$ (see~\cite{MR832433}).
\end{Rem}
\noindent In what follows, $C$ denotes various positive constants that
are independent of $L$, $a$ and $b$.\\
We prove
\begin{Le}
  \label{le:frakN:1} Let $\varphi:\R_+\to\R_+$ be a non increasing
  function such that, for all $l\in\N$, one has $\varphi(l)\le
  1/2$. Then,
  \begin{equation}
    \label{eq:Nfrak-g6}
    \Vert\mathfrak N(L,\varphi,\cdot,\cdot)\Vert_1\leq C\quad\forall
    L\in\N \quad\Longleftrightarrow\quad
    \sum_{N\geq 1}\varphi^6(N)<\infty.
  \end{equation} 
\end{Le}
\noindent and
\begin{Le} \label{le:frakN:2} Let $\varphi:\R_+\to\R_+$ be a non
  increasing function satisfying~(\ref{eq:g-add-cond}).  If
  $\sum_{N\geq 1}\varphi^6(N)$ diverges, then, for all $L\in\N$,
  \begin{equation*}
    \Vert\mathfrak N(L,\varphi,\cdot,\cdot)\Vert_2=
    (1+\delta(L))\,\Vert\mathfrak N(L,\cdot,\cdot)\Vert_1\,
  \end{equation*}
  with $\delta(L)\to 0$ as $L\to\infty$.
\end{Le}
\noindent We prove these two lemmas in the
sections~\ref{sec:analys-dynam-syst},~\ref{sec:NF1}
and~\ref{sec:proof-lemma-refl-1}. We now use them derive
Theorem~\ref{thr:2}.
\subsubsection{The proof of the implication ``$ \Longleftarrow $''
  in~\eqref{eq:3}}
\label{sec:proof-implication-}
In this part of the proof, we choose $\varphi(l)=g(l)$.\\
Note that $\D\Vert\mathfrak N(L,\cdot,\cdot)\Vert_1=\sum_{l=0}^L
m(K_l)$ where
\begin{equation*}
  K_l=\{(a,b);\  \sqrt[4]{a_l}\le\varphi(l)\text{ and }\sqrt{|b_l|}\leq
  \varphi(l)\}.
\end{equation*}
Therefore, Lemma~\ref{le:frakN:1} implies that $\D\sum_{l=0}^\infty
m(K_l)<\infty$.  Therefore, by the Borel-Cantelli lemma, for almost
all $(a,b)\in K$, only a finite number of the
conditions~(\ref{eq:trajectory-conditions}) is satisfied along
$\gamma(a,b)$. Denote the set of such ``good'' $(a,b)$ by $G$. \\
Now, pick $(a,b)\in G$. Let $L_0\in\N$ be large enough so that either
$\sqrt[4]{a_l}\geq g(l)$ or $\sqrt{|b_l|}\geq g(l)$ for all $l\geq
L_0$.  Pick an $L\geq L_0$. Using Proposition~\ref{pro:est-s}, we get
\begin{equation*}
  \max_{N^-(L)\leq N\leq N^+(L)}g(\ln N)\,\frac{|S(N,a,b)|}{\sqrt{N}}\leq 
  C \frac{g(\ln N^-(L))}{g(L)}
\end{equation*}
as $g$ is a non increasing function. And now, as $g$ is a non
increasing function, the implication ``$\Longleftarrow$'' follows from
\begin{Le}
  \label{le:N-L}
  For almost all $0<a<1$, when $ L\to\infty$, one has
  \begin{equation*}
    \ln N^\pm(L)= L(A+o(1))
  \end{equation*}
  where
  \begin{equation}
    \label{eq:4}
    A=\frac1{\ln 2}\int_0^1\frac{\ln(1/a)\,da}{1+a}>1.
  \end{equation}
\end{Le}
\begin{proof}[Proof of Lemma~\ref{le:N-L}]
  Let $a\not\in\Q$. Lemma~\ref{le:1} implies that
  \begin{equation*}
    \frac{\ln(N^-(L))}L=\frac1L\,\sum_{l=0}^{L-1}\ln (1/a_l)\,\,+O(1/L),\quad 
    L\to\infty.
  \end{equation*}
  Recall that the Gauss map $a\to\{1/a\}$ on $(0,1)$ is ergodic, and
  that its invariant measure is $\frac{da}{\ln 2\,(1+a)}$
  (see~\cite{MR832433}).  Therefore, by the Birkhoff-Khinchin Ergodic
  Theorem (\cite{MR832433}), for almost all $a\in(0,1)$, the limit $\D
  \lim_{L\to\infty}\frac1L\,\sum_{l=0}^{L-1}\ln (1/a_l)$ exists and is
  equal to $A$ defined in~\eqref{eq:4}. This completes
  the proof of the asymptotics of $\ln N^-$.\\
  Integrating by parts, we get
  \begin{equation*}
    A=\frac1{\ln 2}\int_0^1\frac{\ln(1+a)}a da\ge
    \frac1{\ln 2}\int_0^1 \left(1-\frac{a}2\right)da=\frac3{4\ln 2}>1.  
  \end{equation*}
  Finally, the asymptotics of $N^+$ follows from~(\ref{Npm}) and the
  asymptotics of $N^-$. This completes the proof of
  Lemma~\ref{le:N-L}.
\end{proof}
\noindent This completes the proof of the implication
``$\Longleftarrow$'' in~\eqref{eq:3}.
\subsubsection{The proof of the implication ``$ \Longrightarrow $''
  in~\eqref{eq:3}}
\label{sec:proof-implication--1}
It suffices to prove that, for almost all $(a,b)\in K$, one has
\begin{equation}
  \label{eq:divergence}
  \sum_{N\geq 1}g^6(N)=+\infty\Longrightarrow
  \limsup_{N\to+\infty}\left(g(\ln N)\,\frac{|S(N,a,b)\,|}
    {\sqrt{N}}\right)=+\infty.
\end{equation}
Let $A$ be the constant defined in~(\ref{eq:4}).  We choose
$\varphi:\R_+\to\R_+$ so that
\begin{itemize}
\item $\D\sum_{l=1}^{\infty}\varphi^6(l)=+\infty$;
\item $r(x):=\varphi(x)/g(2Ax)$ be a monotonously decreasing function;
\item $\D\lim_{x\to\infty}r(x)=0$;
\item $\varphi(x)\leq 1/2$.
\end{itemize}
\begin{Rem} The third and the forth conditions guarantee that
  $\varphi$ satisfies the conditions~(\ref{eq:g-add-cond}).
\end{Rem}
The existence of such a function $\varphi$ follows from
\begin{Le}
  \label{le:varphi-g}
  Let $f:[0,+\infty)\to \R_+$ be a non increasing function such that
  $\D\sum_{l=1}^\infty f(l)=+\infty$. Then,
  \begin{itemize}
  \item for any $C>0$, one has $\D\sum_{l=1}^\infty f(Cl)=+\infty$.
  \item there exists $u:\,[1,+\infty)\to [0,1]$, a monotonously
    decreasing function, such that $\D\lim_{l\to\infty}u(l)=0$ and the
    series $\D\sum_{l=1}^\infty u(l)f(l)$ diverges.
  \end{itemize}
\end{Le}
\begin{proof}[Proof of Lemma~\ref{le:varphi-g}]
  The first statement follows from the fact that for any positive
  valued monotonously non increasing function the series
  $\D\sum_{l=1}^\infty f(l)$ and the integral $\D\int_1^\infty f(x)dx$
  diverge  simultaneously.\\
  To prove the second statement, we pick $0<\alpha<1$ and define
  \begin{equation*}
    u(x)=\alpha\,\left(\int_{0}^x f(x)\,dx\right)^{\alpha-1}.  
  \end{equation*}
  Clearly, $u:\ [1,+\infty)\to \R_+$ is monotonously decreasing, and
  $u(x)$ tends to zero as $x$ tends to infinity. Furthermore, one has
  $\D\int_1^{+\infty} u(x) f(x) dx=+\infty$.  Finally, to satisfy the
  condition $u(x)\leq 1$, it suffices to choose the constant $\alpha$
  small enough. This completes the proof of the second statement.\\
  The proof of Lemma~\ref{le:varphi-g} is complete.
\end{proof}
\noindent Using Lemmas~\ref{le:frakN:1} and~\ref{le:frakN:2}, we now
prove
\begin{Le}
  \label{le: infty-of-conditions}
  There exists a set $B\subset K$ such that $m(B)=1$, and that, for
  all $(a,b)\in B$, there is an infinite sub-sequence of
  conditions~(\ref{eq:trajectory-conditions}) that are satisfied along
  $\gamma(a,b)$.
\end{Le}
\begin{proof}
  We shall use the
  \begin{Le}
    \label{le:Zy-Po}
    Let $K$ be as defined above. Let $\mu$ be a probability measure on
    $K$ and pick $f:K\to\R_+$. Assume that, for some positive constant
    $c$, one has
    \begin{equation*}
      c \Vert f\Vert_{L^2(K,\mu)}\leq \Vert f\Vert_{L^1(K,\mu)}.  
    \end{equation*}
    Then, for any $0<d<c$, one has
    \begin{equation*}
      \mu\left(\, (x,y)\in K\,:\,\, f(x,y)>d\,\Vert
        f\Vert_2\right)\geq (c-d)^2.    
    \end{equation*}
  \end{Le}
  \noindent This actually is a version of the Zygmund-Polya
  Lemma. When $\mu$ is the Lebesgue measure, its proof can be found
  for example in~\cite{MR0306130} (Lemma 2, chapter 7). The same proof
  works in our case.\\
  Pick $\varepsilon\in(0,1/2)$. By Lemma~\ref{le:frakN:2}, for
  sufficiently large $L$, we get
  \begin{equation*}
    (1-\varepsilon)\Vert\mathfrak
    N(L,\cdot,\cdot)\Vert_2\leq \Vert\mathfrak N(L,\cdot,\cdot
    )\Vert_1.
  \end{equation*}
  For such $L$, by Lemma~\ref{le:Zy-Po}, one has
  \begin{equation*}
    m\left(\{ (a,b)\in K\,:\,\, \mathfrak
      N(L,a,b)> \varepsilon\,\Vert\mathfrak
      N(L,\cdot,\cdot)\Vert_1\}\right)\ge (1-2\varepsilon)^2.  
  \end{equation*}
  In view of Lemma~\ref{le:frakN:1}, this implies that the measure of
  the set of $(a,b)$ for which $\mathfrak N(L,a,b)\to+\infty$ as
  $L\to\infty$ is bounded from below by $1-2\varepsilon$. As
  $\varepsilon>0$ can be taken arbitrarily small, this proves
  Lemma~\ref{le: infty-of-conditions}.
\end{proof}
\noindent Now, pick $(a,b)\in B$. There are infinitely many $l$ for
which condition~(\ref{eq:trajectory-conditions}) is satisfied along
$\gamma(a,b)$.  Assume that $L$ is one of them. Using
Proposition~\ref{pro:est-s}, as $g$ is non increasing, we get
\begin{equation*}
  \max_{N^-(L)\leq N\leq N^+(L)}g(\ln N)\,\frac{|S(N,a,b)|}{\sqrt{N}}\geq 
  C \frac{g(\ln N^+(L))}{\varphi(L)}.
\end{equation*}
Combined with Lemma~\ref{le:N-L}, this implies that, for $L$
sufficiently large,
\begin{equation*}
  \max_{N^-(L)\leq N\leq N^+(L)}g(\ln
  N)\,\frac{|S(N,a,b)|}{\sqrt{N}}\geq C\frac{g(2AL)}{\varphi(L)}.
\end{equation*}
For our choice of $\varphi$, the right hand side is equal to $1/r(L)$,
and so, tends to $+\infty$ as $L\to \infty$. This
yields~(\ref{eq:divergence}) and completes the proof of
Theorem~\ref{thr:2}.\qed
\subsection{Analysis of the dynamical system: an invariant family of
  densities}
\label{sec:analys-dynam-syst}
Let $(a_L,b_L)$ be related to $(a,b)$ by~(\ref{eq:als})
and~(\ref{eq:bls}).  In the next subsections, for a fixed $a$, we
study integrals of the form $\int_{-1/2}^{1/2} g(b_L(a,b))\,f(b)\,db$,
where $f(\cdot)$ is considered as a density of a measure. We change
the variable $b$ to $b_L$ to get
\begin{equation*}
  \int_{-1/2}^{1/2} g(b_L)\,f(b)\,db=
  \int_{-1/2}^{1/2} g(b_L) (P_{a_{L-1}}\dots
  P_{a_1}P_af)(b_L) db_L,
\end{equation*}
where
\begin{equation}
  \label{eq:perron-frobenius-fibre}
  \begin{split}
    (P_{a_l}f)(b)=a_l\sum_{m\in\Z\,:\, -1/2<b(m)\leq 1/2}f(b(m))\\
    \text{and  }\quad b(m):=a_l\,(-b+[1/a_l]/2+m).
  \end{split}
\end{equation}
The operator $P_{a_l}$ is the Perron-Frobenius operator of the map
acting on $(-1/2,1/2]$ defined in~(\ref{eq:bls}). In the present
section, we describe a family of densities $f(\cdot)$ invariant under
the cocycle $(a,f(\cdot))\mapsto (\{1/a\},\,P_af(x,\cdot))$ and study
properties of this family.\\
Fix $0<a<1$ and pick $A\geq0$, $B\geq 0$ such that
\begin{equation}
  \label{eq:AB-measure}
  a A+(1-a) B=1.
\end{equation}
The function
\begin{equation}
  \label{eq:density}
  f(b\,|\,a,A,B)=\left\{ \begin{array}{ll}
      A, &  \text{ if } |b|< a/2\\
      B, &  \text{ if } |b|> a/2
    \end{array}\right.
\end{equation}
is the density of a probability measure on  $(-1/2,1/2]$.\\
Our central observation is
\begin{Th}
  \label{th:inv-fam-den}
  Fix $a\in(0,1)$ and choose $A$ and $B$ as above. Then
  \begin{equation}
    \label{eq:Paf}
    P_a f(\cdot\,|\,a,A,B)=f(\cdot\,|\,a_1,A_1,B_1),
  \end{equation}
  where $a_1$ is related to $a$ by~(\ref{eq:ab}), and
  \begin{equation}
    \label{eq:A1B1-AB}
    \begin{pmatrix} A_1\\ B_1\end{pmatrix}=S(a)\,
    \begin{pmatrix} A\\ B\end{pmatrix},\quad
    S(a)=\begin{pmatrix} a & 1-aa_1\\
      a & 1-a-aa_1 \end{pmatrix}.
  \end{equation}
  In addition, one has
  \begin{equation}
    \label{eq:FP:mesure-preserving}
    a_1A_1+(1-a_1)B_1=aA+(1-a)B=1.
  \end{equation}
\end{Th}
\begin{proof}
  Represent $a$ in the form $a=\frac1{N+a_1}$ where $N=[1/a]$ and
  $a_1=\{1/a\}$.  Assume that $N$ is even, i.e.,
  \begin{equation}
    \label{eq:n-even}
    a=\frac1{2n+a_1},\quad\quad  n\in\N,\quad 0\leq a_1<1.
  \end{equation}
  Then, the general formula~(\ref{eq:perron-frobenius-fibre}) can be
  rewritten in the form
  \begin{equation}
    \label{eq:perron-frobenius-fibre-even}
    (P_af)(b)=a\,\cdot\,\left\{\begin{array}{ll}
        \D\sum_{m=-n+1}^{n} f((m-b)a),& \text{ if } b>a_1/2,\\ 
        \D\sum_{m=-n}^{n} f((m-b)a),& \text{ if } |b|\leq a_1/2,\\ 
        \D\sum_{m=-n}^{n-1} f((m-b)a),& \text{ if } b<-a_1/2.
      \end{array}\right.
  \end{equation}
  So, applying $P_a$ to $f(\cdot\,|\,a,A,B)$, and assuming that
  $a_1/2<b_1<1/2$, we get
  \begin{align*}
    (P_af(\cdot\,|\,a,A,B))(b_1)&=a\,\left(
      \D\sum_{m=1}^{n} f((m-b_1)a|a,A,B)\right.\\
    &\hspace{-1cm}
    +f(-b_1a|a,A,B)+\left.\D\sum_{m=-n+1}^{-1} f((m-b_1)a|a,A,B)\right)\\
    &=a(nB+A+(n-1)B)=a(A+(2n-1)B)\\
    &\hspace{.5cm}=aA+(1-a-aa_1)B
  \end{align*}
  as $0<a<1$.\\
  As $f(.|a,A,B)$ is even, we get the same result for
  $-1/2<b_1<-a_1/2$. In the same way as above, we compute
  $(P_af(\cdot\,|\,a,A,B)(a,b_1)=aA+(1-aa_1)B$ for $|b_1|<a_1/2$.\\
  The thus obtained formulas imply~(\ref{eq:Paf})
  and~(\ref{eq:A1B1-AB}) when  $[1/a]$ is even. \\
  The case of odd $[1/a]$ is treated analogously to the case of even $[1/a]$.\\
  Finally, using~(\ref{eq:A1B1-AB}), we get
  \begin{align*}
    a_1A_1+(1-a_1)B_1&=a_1(aA+(1-aa_1)B)\\ &\hspace{2cm}+(1-a_1)(aA+(1-a-aa_1)B)\\
    &=aA+(1-a)B
  \end{align*}
  which proves~(\ref{eq:FP:mesure-preserving}) as $A$ and $B$
  satisfy~(\ref{eq:AB-measure}). This completes the proof of
  Theorem~\ref{th:inv-fam-den}.
\end{proof}
\noindent We now analyze the properties of the
transformation~(\ref{eq:A1B1-AB}).  Let $a\in(0,1)\setminus\Q$.
Consider the sequence $a_0,a_1,a_2,\,\dots $ defined
by~(\ref{eq:als}).  We prove
\begin{Le}
  \label{le:AlBl}
  Pick $l>1$. One has
  \begin{equation*}
    P_{a_{l-1}}P_{a_{l-2}}\dots P_{a_1}P_{a} f(\cdot\,|\,a,A,B)=f(\cdot\,|\,a_l,A_l,B_l),
  \end{equation*}
  where
  \begin{gather}
    \label{eq:Bl-B}
    B_l=1-\sum_{m=0}^{l-2}(-1)^m\prod_{n=l-m}^{l}a_{n}a_{n-1}+
    (-1)^l\prod_{n=1}^{l}a_{n}a_{n-1}\, B,\\
    \label{eq:Al-B}
    A_l=B_l+a_{l-1}B_{l-1}.
  \end{gather}
\end{Le}
\begin{proof}
  Let $A_0=A$ and $B_0=B$.  By Theorem~\ref{th:inv-fam-den}, for $l\in
  \N$,
  \begin{eqnarray}
    \label{eq:Al}
    A_l= a_{l-1}A_{l-1}+(1-a_{l-1}a_l)B_{l-1},\\
    \label{eq:Bl}
    B_l=a_{l-1} A_{l-1}+(1-a_{l-1}-a_la_{l-1})B_{l-1}.
  \end{eqnarray}
  Subtracting~(\ref{eq:Bl}) from~(\ref{eq:Al}), we prove
  (\ref{eq:Al-B}).  Furthermore, substituting into~(\ref{eq:Bl}) with
  $l$ replaced by $l+1$ the value of $A_l$ given by~(\ref{eq:Al-B}),
  we get
  \begin{equation*}
    B_{l+1}=(1-a_{l+1}a_{l})B_{l}+a_la_{l-1}B_{l-1},\quad  \forall l\in\N.
  \end{equation*}
  This implies that
  \begin{equation*}
    B_{l+1}+a_{l+1}a_{l}B_{l}=B_{l}+a_la_{l-1}B_{l-1},\quad  \forall l\in\N.
  \end{equation*}
  Now, for $l=1$, equation~(\ref{eq:Bl}) implies that
  \begin{equation*}
    B_1+a_1a_0B_0=aA_0+(1-a)B_0=1.
  \end{equation*}
  This formula and the previous equation for $\{B_l\}_{l\in\N}$ imply
  that
  \begin{equation*}
    B_{l}=1-a_{l}a_{l-1}B_{l-1},\quad  \forall l\in\N.
  \end{equation*}
  This relation allows to express $B_l$ directly in terms of $B_0=B$,
  and one obtains~(\ref{eq:Bl-B}). This completes the proof of
  Lemma~\ref{le:AlBl}.
\end{proof}
\noindent To complete this section, we discuss another family of
densities $f(\cdot\,|\,a,M)$, $M\in\N$, such that
$P_af(\cdot\,|\,a,M)=f(\cdot\,|\,a_1,A,B)$.  We prove
\begin{Le}
  \label{le:second-family}
  For $a\in(0,1)$ and $M\in\N$ satisfying,
  \begin{equation}
    \label{eq:M-new-family}
    M\le\left\{\begin{array}{ll} \frac12\,\left[\frac1a\right] & 
        \text{ if } [1/a] \text{ is even,}\\
        \\
        \frac12\,\left[\frac1a +1\right] & 
        \text{ if } [1/a] \text{ is odd.}
      \end{array}\right.
  \end{equation}
  Let
  \begin{equation}
    \label{eq:second-family}
    f(b|a,M)=\left\{\begin{array}{ll}
        \D\frac{\chi(\,|b|\leq a(M-a_1/2)\,)}{a(2M-a_1)} & 
        \text{ if } [1/a] \text{ is even,}\\
        \\
        \D\frac{\chi(\,|b|\leq a(M-1/2-a_1/2)\,)}{a(2M-1-a_1)} & 
        \text{ if } [1/a] \text{ is odd.}
      \end{array}\right.
  \end{equation}
  Then,
  \begin{equation}
    \label{eq:new-family-trans}
    P_a f(\cdot\,|a,M)=f(\cdot\,|a_1,A_1,B_1),
  \end{equation}
  and
  \begin{equation}
    \label{eq:new-family-trans:1} 
    \text{if } M>1,\text{ then } A_1,B_1=1+O(1/M),
  \end{equation}
  the error estimate being uniform in $a$.
\end{Le}
\begin{proof}
  Assume that $[1/a]$ is even. In the sums in the right hand side
  of~(\ref{eq:perron-frobenius-fibre-even}), only the terms with
  $-M+a_1/2+b_1\leq m\leq M-a_1/2+b_1$ are non zero.  So, for
  $a_1/2<b_1<1/2$, we get
  \begin{equation*}
    \begin{split}
      (P_af(\cdot\,|\,a,M))(b_1)&=a\,
      \D\sum_{m=-M+1}^{M} f((m-b_1)a|a,M)\\
      &=\frac{2M}{2M-a_1}=1+\frac{a_1}{2M-a_1}.
    \end{split}
  \end{equation*}
  And, for $0<b_1<a_1/2$, we obtain
  \begin{equation*}
    \begin{split}
      (P_af(\cdot\,|\,a,M))(a,b_1)&=a\,
      \D\sum_{m=-M+1}^{M-1} f((m-b_1)a|a,M)\\
      &=\frac{2M-1}{2M-a_1}=1-\frac{1-a_1}{2M-a_1}.
    \end{split}
  \end{equation*}
  In the case of negative $b_1$, we obtain the same formulas as for
  $-b_1$. This implies~\eqref{eq:new-family-trans} with
  \begin{equation*}
    A_1=1-\frac{1-a_1}{2M-a_1}\quad\text{ and }\quad
    B_1=1+\frac{a_1}{2M-a_1}.
  \end{equation*}
  As $0<a_1<1$ and $M\geq 1$, we see that $A_1,B_1=1+O(1/M)$. \\
  This completes the proof of Lemma~\ref{le:second-family} for $[1/a]$
  even. To complete the proof of Lemma~\ref{le:second-family}, the
  case of odd $[1/a]$ is treated similarly.
\end{proof}
\subsection{Proof of Lemma~\ref{le:frakN:1}}
\label{sec:NF1}
By~(\ref{eq:N-chi}),
\begin{equation}
  \label{eq:NF1}
  \Vert\mathfrak N(L,\cdot,\cdot)\Vert_1=\sum_{l=0}^L\int_K
  \chi(a_l\le\varphi^4(l))\,
  \chi(b_l\le\varphi^2(l))\,\frac{da\,db}{\ln2\,(1+a)},
\end{equation}
where $(a_l,b_l)$ are related to $(a,b)$ by~(\ref{eq:als})
and~(\ref{eq:bls}).  To transform the right hand side
of~(\ref{eq:NF1}), we first use Fubini's theorem and then, for fixed
$a$, we perform the change of variable $b\to b_l$.  As $f(b|a,1,1)=1$,
Lemma~\ref{le:AlBl} implies that
\begin{equation}
  \label{eq:NF2}
  \Vert\mathfrak
  N(L,\cdot,\cdot)\Vert_1=\frac1{\ln2}\sum_{l=0}^L\int_0^1
  \frac{\chi(a_l\le\varphi^4(l))\cdot I(l) \,da}{1+a}, 
\end{equation}
where
\begin{equation*}
  I(l):=\int_{-1/2}^{1/2}
  \chi(|b_l|\le\varphi^2(l))\,f(b_l|a_l,A _l,B_l)\,db_l,   
\end{equation*}
the coefficients $A_l$ and $B_l$ being defined by~(\ref{eq:Al-B})
and~(\ref{eq:Bl-B}) with $B_0=1$.\\
Recall that $\varphi_l<1/2$.\\
Let us study $I(l)$ under the condition $a_l\leq \varphi^4(l)$.
Using~(\ref{eq:density}), we compute
\begin{align}
  \label{eq:NF3}
  I(l)&=2\left(A_l\int_{0}^{a_l/2}+
    B_l\int_{a_l/2}^{1/2}\right) \chi(b_l\le\varphi^2(l))\,db_l\\
  \nonumber
  &=(A_l\,a_l+B_l\,(2\varphi^2(l)-a_l))=(a_l\,a_{l-1}\,B_{l-1}
  +2B_l\,\varphi^2(l)),
\end{align}
where, in the second step, we used the inequalities $a_l/2\le
\varphi^4(l)/2<\varphi^2(l)$ and $\varphi^2(l)<1/2$ which follows from
$\varphi(l)<1/2$, and, in the last step, we
used~(\ref{eq:Al-B}). \\
Note that it follows from estimate~(\ref{eq:two-a-l}) and
formula~(\ref{eq:Bl-B}) with $B_0=1$ that, for all $l\geq 0$, one has
$1/2< B_l< 1$. Therefore,
\begin{equation}
  \label{eq:FN1:2}
  \varphi^2(l)< 2B_l\,\varphi^2(l)<I(l)<a_l+2\varphi^2(l)<3\varphi^2(l).
\end{equation}
Let us now turn to the study of $\Vert\mathfrak N(L,\cdot,\cdot
)\Vert_1$.  As the density $\frac1{\ln 2\,(1+a)}$ is inva\-riant with
respect to the Gauss transformation $a\to \{1/a\}$, one computes
\begin{equation}
  \label{eq:FN1:3}
  \int_0^1\frac{\chi(a_l\le\varphi^4(l))\,da}{1+a}=
  \int_0^1\frac{\chi(a_l\le\varphi^4(l))\,da_l}{1+a_l}=\ln(1+\varphi^4(l)).
\end{equation}
The inequality~(\ref{eq:FN1:2}) and the equality~(\ref{eq:FN1:3})
imply that
\begin{equation*}
  \frac1{\ln2} \sum_{l=0}^L \ln(1+\varphi^4(l))\varphi^2(l)\leq 
  \Vert\mathfrak N(L,\cdot,\cdot)\Vert_1\leq 
  \frac3{\ln2} \sum_{l=0}^L \ln(1+\varphi^4(l))\varphi^2(l).
\end{equation*}
This implies~(\ref{eq:Nfrak-g6}), hence, completes the proof of
Lemma~\ref{le:frakN:1}.  \qed
\subsection{Proof of Lemma~\ref{le:frakN:2}}
\label{sec:proof-lemma-refl-1}
We now assume that $\D\lim_{l\to\infty}\varphi(l)=0$. This enables us
to get more precise estimates for $\Vert\mathfrak
N(L,\cdot,\cdot)\Vert_1$ in subsection~\ref{sec:estim-vertm-nl}. In
subsection~\ref{sec:estim-vertm-nl-1}, using these estimates, we
approximate $\Vert\mathfrak N(L,\cdot,\cdot)\Vert_2$ with
$\Vert\mathfrak N(L,\cdot,\cdot)\Vert_1$ and, thus, prove
Lemma~\ref{le:frakN:2}.\\
Below, $C$ denotes positive constants independent of $a,b,L$ and other
variables (e.g., indices of summation). Moreover, when writing
$f=O(g)$, we mean that $|f|\leq C|g|$.
\subsubsection{Precise estimates for $\Vert\mathfrak
  N(L,\cdot,\cdot)\Vert_1$}
\label{sec:estim-vertm-nl}
Recall that, in formula~(\ref{eq:NF3}), one has $a_l\leq\varphi^4(l)$
and $B_l$ is computed by~(\ref{eq:Bl-B}) with
$B=1$. Formula~(\ref{eq:Bl-B}) with $B=1$ implies that $1/2<B_m<1$ for
all $m$. Moreover, as $a_l\leq\varphi^4(l)$ and $\varphi(l)$ is small,
we can write $B_l=1+O(\varphi^4(l))$.  So, we replace~(\ref{eq:NF3})
with
\begin{equation}
  \label{eq:Ila-new}
  I(l)=2\varphi^2(l)(1+O(\varphi^2(l))).
\end{equation}
This and~(\ref{eq:NF1}) imply that
\begin{equation}
  \label{eq:NF1new}
  \Vert\mathfrak N(L,\cdot,\cdot)\Vert_1=\sum_{l=0}^L J(l) \text{
    where }
  J(l)=\frac2{\ln 2}\varphi^6(l)\,(1+O(\varphi^2(l))).
\end{equation}
That is the formula that we need to estimate $\Vert\mathfrak
N(L,\cdot,\cdot)\Vert_2$.
\subsubsection{Estimates for $\Vert\mathfrak N(L,\cdot,\cdot)\Vert_2$}
\label{sec:estim-vertm-nl-1}
Using~(\ref{eq:N-chi}), we get
\begin{equation}
  \label{eq:FN2:1}
  \begin{split}
    \Vert&\mathfrak N(L,\cdot,\cdot)\Vert^2_2=\Vert\mathfrak
    N(L,\cdot,\cdot)\Vert_1+\\
    &+\frac2{\ln 2}\sum_{0\leq l<m\leq L}
    \int_0^1\frac{da}{1+a}\chi(a_l\leq \varphi^4(l)) \chi(a_m\leq
    \varphi^4(m))\,I(l,m)
  \end{split}
\end{equation}
where
\begin{equation}
  \label{eq:FN2:2}
  I(l,m)=\int_{-1/2}^{1/2} \chi(|b_l|\leq \varphi^2(l))
  \chi(|b_m|\leq \varphi^2(m))\,db.
\end{equation}
The central ingredient for the proof of Lemma~\ref{le:frakN:2} is
\begin{Le}
  \label{le:Ilm} Let $l<m$.  If $a_l\leq \varphi^4(l)$ and $a_m\le
  \varphi^4(m)$, then
  \begin{equation}
    \label{eq:FN2:Ilm}
    I(l,m)=4\varphi^2(l)\varphi^2(m)
    \left(1+O(\varphi^2(l))\right).
  \end{equation}
\end{Le}
\begin{proof}
  The analysis of the integral $I(l,m)$ begins as the analysis of the
  integral $I(l)$ in the previous section, and one easily computes
  \begin{equation*}
    \begin{split}
      I(l,m)&=\int_{-1/2}^{1/2} \chi(|b_l|\leq \varphi^2(l))
      \chi(|b_m|\leq \varphi^2(m)) f(b_l|a_l,A_l,B_l) db_l\\
      &=\int_{|b_l|< \varphi^2(l)} \chi(|b_m|\leq \varphi^2(m))
      f(b_l|a_l,A_l,B_l) db_l\\
      &=\left(A_l\int_{ |b_l|< a_l/2} +B_l\int_{
          a_l/2<|b_l|<\varphi^2(l)}\right)\, \chi(|b_m|\leq
      \varphi^2(m)) db_l,
    \end{split}
  \end{equation*}
  and
  \begin{equation}
    \label{eq:I1I2}
    I(l,m)=a_la_{l-1}B_{l-1}I_1(l,m)+B_l I_2(l,m),
  \end{equation}
  where $B_l$ and $A_l$ are computed by~(\ref{eq:Bl-B}) with $B=1$,
  and we have set
  \begin{gather}
    \nonumber I_1(l,m)=\frac1{a_l}\int_{|b_l|<a_l/2}\chi(|b_m|\leq
    \varphi^2(m)) db_l,\\
    \label{eq:FN2:4}
    I_2(l,m)=\int_{|b_l|<\varphi^2(l)}\chi(|b_m|\leq \varphi^2(m))
    db_l.
  \end{gather}
  Estimate the integral $I_1(l,m)$. Therefore, we use
  Lemma~\ref{le:AlBl} with the sequence $(a_j)_{j\geq l}$ instead of
  the sequence $(a_j)_{j\geq 0}$. We compute
  \begin{equation*}
    \begin{split}
      I_1(l,m)&=\int_{-1/2}^{1/2}\chi(|b_m|\leq \varphi^2(m))
      f(b_l|a_l,1/a_l,0)\,db_l\\&=\int_{-1/2}^{1/2}\chi(|b_m|\leq
      \varphi^2(m)) f(b_m|a_m,\tilde A_m,\tilde B_m)\,db_m\\
      &= a_m (\tilde A_m-\tilde B_m)+2\tilde B_m\varphi^2(m),
    \end{split}
  \end{equation*}
  where $\tilde A_m$ and $\tilde B_m$ are computed in terms of $\tilde
  A_l=1/a_l$ and $\tilde B_l=0$ by formulas~(\ref{eq:Al})
  and~(\ref{eq:Bl}). Formula~(\ref{eq:Bl-B}) implies that $\tilde
  B_{m-1}\leq 1$, and $\tilde B_m=1+O(a_m)$. These observations
  and~(\ref{eq:Al-B}) lead to the estimate
  \begin{equation}\label{eq:1111}
    I_1(l,m)=O(\varphi^2(m)).
  \end{equation}
  To compute the integral $I_2(l,m)$, we use
  Lemma~\ref{le:second-family} with $a$ and $a_1$ replaced with
  $a_l$ and $a_{l+1}$.\\
  Consider the case when $[1/a_l]$ is even. Choose an integer $M$ so
  that
  \begin{equation}
    \label{eq:FN2:5}
    0\leq a_l(M-a_{l+1}/2)-\varphi^2(l)< a_l.
  \end{equation}
  As $a_l\leq \varphi^4(l)$ and $\varphi(l)<1/2$, one has
  \begin{equation}
    \label{eq:FN2:6}
    M\ge 1/\varphi^2(l)>4.
  \end{equation}
  The definition of $I_2(l,m)$,~\eqref{eq:new-family-trans}
  and~\eqref{eq:new-family-trans:1} yield
  \begin{equation*}
    \begin{split}
      I_2(l,m)&\leq\int_{|b_l|<a_l(M-a_{l+1}/2)}\chi(|b_m|\leq
      \varphi^2(m))\,db_l
      \\&=2a_l(M-a_{l+1}/2)\int_{-1/2}^{1/2}\chi(|b_m|\leq
      \varphi^2(m))\,f(b_l|a_l,M)\,db_l\\
      &=2a_l(M-a_{l+1}/2)\int_{-1/2}^{1/2}\chi(|b_m|\leq
      \varphi^2(m))\, f(b_{l+1}|a_{l+1},A,B)\,db_{l+1}
    \end{split}
  \end{equation*}
  with $A,B=1+O(1/M)$. Moreover, in view of~(\ref{eq:FN2:6}), one has
  \begin{equation}
    \label{eq:FN2:7}
    A,B=1+O(1/M)=1+O(\varphi^2(l)).
  \end{equation}
  If $m=1+l$, we compute
  \begin{equation*}
    I_2(l,m)=2a_l(M-a_{m}/2)\,(a_m(A-B)+2B\varphi^2(m));
  \end{equation*}
  using~(\ref{eq:FN2:7}) and~(\ref{eq:FN2:5}), we finally obtain
  \begin{equation}
    \label{eq:FN2:8}
    I_2(l,m)\leq 4\varphi^2(l)\varphi^2(m)\,(1+O(\varphi^2(l))).
  \end{equation}
  If $m>l+1$, in the last integral for $I_2(l,m)$, we change the
  variable $b_{l+1}$ to $b_m$ and get
  \begin{equation*}
    \begin{split}
      I_2(l,m)&\leq
      2a_l(M-\frac{a_{l+1}}2)\int_{-1/2}^{1/2}\chi(|b_m|\le
      \varphi^2(m))\,
      f(b_{m}|a_{m},\tilde{\tilde A}_m,\tilde{\tilde B}_m)db_{l+1}\\
      &=2a_l(M-a_{l+1}/2)\left(a_m(\tilde{\tilde A}_m-\tilde{\tilde
          B}_m)+ 2\tilde{\tilde B}_m \varphi^2(m)\right)
    \end{split}
  \end{equation*}
  where $\tilde{\tilde A}_m$ and $\tilde{\tilde B}_m$ are obtained
  from $\tilde{\tilde A}_{l+1}=A$ and $\tilde{\tilde B}_{l+1}=B$ by
  formulas~(\ref{eq:Al}) and~(\ref{eq:Bl}). Now,
  using~(\ref{eq:FN2:5}) and Lemma~\ref{le:AlBl} with $(a_j)_{j\geq
    l+1}$ instead of $(a_j)_{j\geq0}$, as $l<m$ and $\varphi$ is non
  increasing, we get
  \begin{equation}
    \label{eq:FN2:9}
    \begin{split}
      I_2(l,m)&\le
      4\varphi^2(l)\varphi^2(m)(1+O(\varphi^2(m))(1+O(\varphi^2(l)))\\
      &\leq 4\varphi^2(l)\varphi^2(m)(1+O(\varphi^2(l))).
    \end{split}
  \end{equation}
  We now complete the proof of Lemma~\ref{le:Ilm}.  First, it follows
  from Lemma~\ref{le:AlBl} that
  \begin{equation}
    \label{eq:7}
    a_l=O(\varphi^4(l)),\quad 
    B_{l-1}\leq 1\quad\text{and}\quad
    B_l=1+O(a_l)=1+O(\varphi^4(l)).
  \end{equation}
  We plug~(\ref{eq:1111}), (\ref{eq:FN2:8}) and~(\ref{eq:FN2:9})
  into~\eqref{eq:I1I2}. Taking into account~(\ref{eq:7}), we
  obtain~(\ref{eq:FN2:Ilm}).  This completes the proof of
  Lemma~\ref{le:Ilm}.
\end{proof}
\noindent We now return to the study of $\|\mathfrak
N(L,\cdot,\cdot)\|_2$. Using well known properties of the Gauss map,
we prove
\begin{Le}
  \label{le:N2-gauss}
  One has
  \begin{equation*}
    \|\mathfrak N(L,\cdot,\cdot)\|_1^2\leq\|\mathfrak
    N(L,\cdot,\cdot)\|_2^2\leq \|\mathfrak
    N(L,\cdot,\cdot)\|^2_1+\|\mathfrak N(L,\cdot,\cdot)\|_1+R_L
  \end{equation*}
  where, for some $C>0$, one has
  \begin{equation*}
    R_L:=
    \D\sum_{m,l=0}^{L}\,\varphi^6(l)\varphi^6(m)\cdot
    O\left(\varphi^2(m)+\varphi^2(l)+e^{-(m-l)/C}\right).
  \end{equation*}
\end{Le}
\begin{proof}
  The lower bound on $\|\mathfrak N(L,\cdot,\cdot)\|_2^2$ is a
  consequence of the Cauchy-Schwarz inequality.\\
  To prove the upper bound, we substitute~(\ref{eq:FN2:Ilm})
  into~(\ref{eq:FN2:1}) to get
  \begin{multline}
    \label{eq:FN2:10}
    \|\mathfrak N(L,\cdot,\cdot)\|_2^2=\|\mathfrak
    N(L,\cdot,\cdot)\|_1\\+ 8\sum_{0\leq l<m\leq L}
    \varphi^2(l)\varphi^2(m)\, P\left(\,a_l\leq \varphi^4(l),\,a_m\le
      \varphi^4(m)\,\right) \left(1+O(\varphi^2(l))\right),
  \end{multline}
  where we have defined
  \begin{equation*}
    P(a_l\leq \alpha, a_m\leq \beta):=\frac1{\ln 2}
    \int_0^1\frac{da}{1+a}\chi(a_l\leq \alpha)
    \chi(a_m\leq \beta)
  \end{equation*}
  i.e. $P(a_l\leq \alpha, a_m\leq \beta)$ is the probability (with
  respect to the invariant measure of the Gauss map) that $a_m<\beta$
  and $a_l<\alpha$. It is controlled by Gordin's Theorem
  (see~\cite{MR0245544}, Theorem 3 and remarks following this
  theorem). By Gordin's Theorem, there exists two constants $A>0$ and
  $\lambda>0$ such that, for all $0\leq l<m<\infty$ and for any
  integer $\alpha>0$ and any real number $\beta>0$, one has
  \begin{multline}
    \label{eq:gauss-mixing}
    |P(a_l\leq 1/\alpha, a_m\leq \beta)-
    P(a_l\leq 1/\alpha)P( a_m\leq \beta)|\\
    \leq A\,P(a_l\leq 1/\alpha)\,P( a_m\leq \beta)\, e^{-\lambda
      (m-l)}
  \end{multline}
  where we have defined
  \begin{equation*}
    P(a_l\leq \alpha):=
    \frac1{\ln 2}\int_0^1\frac{da}{1+a}\chi(a_l\leq \alpha).
  \end{equation*}
  Now, choose a positive integer $s$ so that
  \begin{equation*}
    \frac1{s+1}\leq \varphi^4(l)<\frac1s.
  \end{equation*}
  Note that, as $\varphi(l)<1/2$, such a positive integer exists, and
  that
  \begin{equation}
    \label{eq:FN2:12}
    \frac1s-\varphi^4(l)=O(\varphi^8(l)).
  \end{equation}
  Using~(\ref{eq:gauss-mixing}), we get
  \begin{equation*}
    \begin{split}
      P(\,a_l\leq \varphi^4(l),&\,a_m\leq \varphi^4(m)\,)\leq
      P(\,a_l\leq 1/s,\,a_m\leq \varphi^4(m)\,)\\
      &\leq P(\,a_l\leq 1/s\,)\,P(\,a_m\leq \varphi^4(m)\,)
      (1+Ae^{-\lambda(m-l)}).
    \end{split}
  \end{equation*}
  Using the definition of the invariant measure, we obtain
  \begin{equation*}
    \begin{split}
      P(\,a_m\leq \varphi^4(m)\,)&=P(a\le\varphi^4(m))=
      \frac1{\ln 2}\int_0^{\varphi^4(m)}\frac{da}{1+a}\\
      &=\frac{\ln(1+\varphi^4(m))}{\ln2}=\frac1{\ln
        2}\varphi^4(m)\,(1+O(\varphi^4(m))).
    \end{split}
  \end{equation*}
  In the same way,~(\ref{eq:FN2:12}) yields
  \begin{equation*}
    P(\,a_l\leq 1/s\,)=\frac1{s\,\ln 2}\,\,(1+O(1/s))=
    \frac1{\ln 2}\varphi^4(l)\,(1+O(\varphi^4(l))).  
  \end{equation*}  
  These two results imply that
  \begin{multline*}
    P(a_l\leq \varphi^4(l),\ a_m\leq \varphi^4(m))\\\leq
    \frac1{(\ln2)^2}\,\varphi^4(l)\,\varphi^4(m)
    \left(1+O\left(\varphi^4(l)+\varphi^4(m)+Ae^{-\lambda(m-l)}\right)\right).
  \end{multline*}
  Combining this estimate and~(\ref{eq:FN2:10}),
  recalling~\eqref{eq:NF1new}, we obtain the upper bound on
  $\|\mathfrak N(L,\cdot,\cdot)\|_2^2$ announced in
  Lemma~\ref{le:N2-gauss}. This completes the proof of
  Lemma~\ref{le:N2-gauss}.
\end{proof}
\noindent Now, we can complete the proof of Lemma~\ref{le:frakN:2} by
means of elementary estimates. Recall that by assumption of
Lemma~\ref{le:frakN:2}, $\sum_{l=0}^\infty\varphi^6(l)$
diverges. By~\eqref{eq:NF1new}, this implies that $\D\|\mathfrak
N(L,\cdot,\cdot)\|_1=\sum_{l=0}^LJ(l)\to \infty$ as $L\to\infty$.  So,
to prove that $\|\mathfrak N(L,\cdot,\cdot)\|_2^2=\|\mathfrak
N(L,\cdot,\cdot)\|_1^2(1+o(1))$ when $L\to\infty$, and, thus, to
complete the proof of Lemma~\ref{le:frakN:2}, it suffices to show that
\begin{gather}
  \label{eq:fin1}
  \lim_{L\to\infty}\frac{\sum_{l=0}^L J(l) \varphi^2(l)}
  {\sum_{l=0}^L J(l)}=0,\\
  \label{eq:fin2}
  \lim_{L\to\infty}\frac{\sum_{l,m=0}^L J(l)J(m) e^{-|l-m|/C}}
  {\left(\sum_{l,m=0}^L J(l)\right)^2}=0.
\end{gather}
As $\varphi(l)\to0$ and $\D\sum_{l=0}^LJ(l)\to
\infty$,~(\ref{eq:fin1})
is a standard result of Cesaro convergence.\\
As $J(m)$ is bounded uniformly in $m$,~(\ref{eq:fin2}) follows from
\begin{gather*}
  \frac{\sum_{l,m=0}^L J(l)J(m) e^{-|l-m|/C}} {\left(\sum_{l=0}^L
      J(l)\right)^2}\leq C\,\frac{\sum_{l=0}^L J(l)\sum_{m=0}^L
    e^{-|l-m|/C}} {\left(\sum_{l=0}^L J(l)\right)^2}\le
  \frac{C}{\sum_{l=0}^L J(l)}.
\end{gather*}
This completes the proof of Lemma~\ref{le:frakN:2}
\section{The proof of Theorem~\ref{thr:2a}}
\label{sec:proof-theor-refthr:2}
\noindent Let $g$ be as in Theorem~\ref{thr:2a}. We first prove
\begin{Le}
  \label{thr:1}
  Let $g:\R_+\to\R_+$ be a non increasing function such that
  \begin{equation*}
    \sum_{N\geq 1}g^4(N)<\infty.
  \end{equation*}
  Then, for almost
  all $a\in(0,1)$ and for all $b\in(-1/2,1/2]$, one has
  \begin{equation}
    \label{eq:18}
    \limsup_{N\to+\infty}\left(g(\ln N)\,
      \frac{|S(N,a,b)\,|}{\sqrt{N}}\right)<\infty.
  \end{equation}
\end{Le}
\begin{proof}[Proof of Lemma~\ref{thr:1}]
  As $\D \sum_{N\geq 1}g^4(N)<\infty$, Theorem 30
  of~\cite{MR98c:11008} implies that, for almost all $a\in(0,1)$,
  there exists $L_0\in\N$ such that $a_l\geq g^4(l)$ for all $l\geq
  L_0$. Pick $L\geq L_0$. Using Proposition~\ref{pro:est-s}, we get
  \begin{equation*}
    \max_{N^-(L)\leq N\leq N^+(L)}g(\ln N)\,\frac{|S(N,a,b)|}{\sqrt{N}}\leq 
    C \frac{g(\ln N^-(L))}{g(L)}
  \end{equation*}
  as $g$ is a non increasing function. And now, as $g$ is a non
  increasing function,~\eqref{eq:18} follows from Lemma~\ref{le:N-L}.
  This completes the proof of Lemma~\ref{thr:1}.
\end{proof}
\noindent Now, Theorem~\ref{thr:2a} follows from
\begin{Pro}
  \label{pro:1}
  Let $g: \N\to \R_+$ be a non increasing function such that
  \begin{equation*}
    \sum_{N\geq 1}g^4(N)=\infty.
  \end{equation*}
  Then, for almost
  all $a\in(0,1)$ and all $b\in {\mathcal B}_a$, one has
  \begin{equation*}
    \limsup_{N\to+\infty}\left(g(\ln N)\, 
      \frac{|S(N,a,b)\,|}{\sqrt{N}}\right)=\infty.
  \end{equation*}
\end{Pro}
\noindent Indeed, if $ \sum_{N\geq 1}g^4(N)=\infty $, by
Proposition~\ref{pro:1}, for almost all $a$, as ${\mathcal B}_a$ is dense
in $(-1/2,1/2]$, the set
\begin{gather*}
  \tilde {\mathcal B}_a:=\left\{b\in(-1/2,1/2];\
    \limsup_{N\to+\infty}\left(g(\ln N)\,
      \frac{|S(N,a,b)\,|}{\sqrt{N}}\right)=+\infty\right\}
\end{gather*}
is dense in $(-1/2,1/2]$. As $b\mapsto S(N,a,b)$ is continuous and as
\begin{equation*}
  \tilde {\mathcal B}_a=\bigcap_{K\geq 1}\bigcap_{M\geq
    1}\bigcup_{N\geq M}\left\{b\in(-1/2,1/2];\ 
    g(\ln N)\, \frac{|S(N,a,b)\,|}{\sqrt{N}}> K\right\},
\end{equation*}
$\tilde {\mathcal B}_a$ is a dense $G_\delta$-set. This completes the
proof of Theorem~\ref{th:inv-fam-den} once Proposition~\ref{pro:1} is
proved.
\subsection{Proof of Proposition~\ref{pro:1}}
\label{sec:proof-proposition}
Proposition~\ref{pro:1} follows from
\begin{Le}
  \label{le:2}
  For $(a_0,b_0)$, define the inductive sequence $(a_n,b_n)$ by
  formulas~\eqref{eq:als} and~\eqref{eq:bls}.\\
  Then, for almost every $a$ and all $b\in\mathcal{B}_a$, there exists
  $j_0\geq1$ such that, for $j\geq j_0$, one has
  \begin{equation}
    \label{eq:9}
    b_j\in\left\{0,\frac12,-\frac{a_j}{2}\right\}.
  \end{equation}
\end{Le}
\noindent and
\begin{Pro}
  \label{pro:2}
  Let $g: \N\to \R_+$ be a non increasing function such that
  \begin{equation*}
    \sum_{N\geq 1}g^4(N)=\infty.
  \end{equation*}
  Then, for almost all $a\in(0,1)$ and $b\in\{0,1/2,-a/2\}$, one has
  \begin{equation}
    \label{eq:8}
    \limsup_{N\to+\infty}\left(g(\ln N)\, 
      \frac{|S(N,a,b)\,|}{\sqrt{N}}\right)=\infty.
  \end{equation}
\end{Pro}
\noindent Indeed, let $\mathcal{A}_0$ be the set of total measure of
$a$'s defined by Lemma~\ref{le:2}. For $p\in\N$, let $g_p:\R^+\to\R^+$
be the function $g_p(x)=g(x+p)$. If $\sum_{N\geq 1}g^4(N)=\infty$
then, for any $p\in\N$, one has $\sum_{N\geq 1}g_p^4(N)=\infty$. Let
$\mathcal{A}^p$ be the set of total measure of $a$'s defined by
Proposition~\ref{pro:2} where the function $g$ is replaced by the
function $g_p$.\\
If $G$ denotes the Gauss map (see~(\ref{eq:als})), the set
$\D\mathcal{A}_0\cap\bigcap_{p,l\geq0}G^{-l}(\mathcal{A}^p)$ is of
total measure. For $a$ in this set and $b\in\mathcal{B}_a$, there
exists $j_0$ even such that~(\ref{eq:9}) is satisfied and~(\ref{eq:8})
is satisfied for $(a_{j_0},b_{j_0})$ and $g$ replaced by any
$g_p$. Applying the renormalization formula~(\ref{eq:exa-ren-for})
$j_0$ times, we see that
\begin{equation*}
  S(N,a,b)=C_{j_0}\,S(N_{j_0},a_{j_0},b_{j_0})+O(1)
\end{equation*}
where $\sqrt{a_0\cdots a_{j_0}}|C_{j_0}|=1$ and $N_{j_0}=N_{j_0}(N)$
is defined in~(\ref{eq:ab}) and satisfies $N_{j_0}\sim a_0\cdots
a_{j_0}N$ when $N\to+\infty$. Hence
\begin{equation*}
  \frac{|S(N,a,b)\,|}{\sqrt{N}}\equ_{N\to+\infty}
  \frac{|S(N_{j_0},a_{j_0},b_{j_0})\,|}{\sqrt{N_{j_0}}}.
\end{equation*}
Moreover, for $p_0\geq|\ln\sqrt{a_0\cdots a_{j_0}}|+1$ and $N$
sufficiently large, one has $g_{p_0}(\ln N_{j_0}(N))\leq g(\ln
N)$. Finally, noticing that when $N$ goes to $\infty$ running through
all the integers, $N_{j_0}=N_{j_0}(N)$ does so too, we obtain
\begin{equation*}
  \limsup_{N\to+\infty}\left(g(\ln N)\,
    \frac{|S(N,a,b)\,|}{\sqrt{N}}\right)\geq\limsup_{N\to+\infty}
  \left(g_{p_0}(\ln N)\,
    \frac{|S(N,a_{j_0},b_{j_0})\,|}{\sqrt{N}}\right)=\infty.
\end{equation*}
So we have proved that Proposition~\ref{pro:2} and Lemma~\ref{le:2}
imply Proposition~\ref{pro:1}.\\[1mm]
Proposition~\ref{pro:2} is proved in
section~\ref{sec:proof-proposition}. We now turn to the proof of
Lemma~\ref{le:2}.
\begin{proof}[Proof of Lemma~\ref{le:2}] Pick $a=a_0\in(0,1)$
  arbitrary and let $b_0\in {\mathcal B}_a$. One can represent $b_0$
  as
  \begin{equation*}
    b_0=\frac12\,(n_0\,a_0-[n_0\,a_0]-\varepsilon_0),\quad n_0\in\Z,\quad
    \varepsilon_0\in\{0,1\}.
  \end{equation*}
  Computing $b_1$ from $b_0$ by formula~\eqref{eq:ab}, one obtains
  \begin{equation}\label{eq:b0b1}
    b_1=\left\{\frac12\,([n_0a_0]+\varepsilon_0)\,a_1+
      \frac12\,\left(([n_0a_0]+\varepsilon_0+1)\,\left[\frac1{a_0}\right]-n_0\right)
    \right\}_0.
  \end{equation}
  Therefore,
  \begin{gather*}
    b_1=b_1(b_0)=\frac12 (n_1\,a_1-[n_1\,a_1]-\varepsilon_1),\\
    n_1=[a_0\,n_0]+\varepsilon_0,\quad\quad \varepsilon_1\in\{0,1\}.
  \end{gather*}
  Hence, we can define $(b_j)_{j\geq0}$ by formula~\eqref{eq:bls} and
  represent it as above as
  \begin{gather*}
    b_j=\frac12\,(
    n_j\,a_j-[n_j\,a_j]-\varepsilon_j),\\
    n_j=[a_{j-1}\,n_{j-1}]+\varepsilon_{j-1},\quad
    \varepsilon_j\in\{0,1\}.
  \end{gather*}
  Note that, if $n_{j-1}\in\{-1,0,1\}$ then
  $n_j\in\{-1,0,1\}$.\\ \\
  Let $\mathcal{Z}_a=\frac12((2\Z+1)a+(2\Z+1))$. We note that
  (see~\eqref{eq:b0b1})
  \begin{equation*}
    b_{j+1}\in {\mathcal Z}_{a_{j+1}}\quad \Longleftrightarrow\quad
    b_j\in {\mathcal Z}_{a_j}.
  \end{equation*}
  So, for $b\in {\mathcal B}_a$, for any $j\geq0$,
  $b_j\not\in{\mathcal Z}_{a_j}$.\\
  Consider now the sequence $(\beta_j)_{j\geq0}$ defined by
  \begin{equation*}
    \beta_0=|n_0|,\quad \beta_{j+1}=a_j\,\beta_j+1\text{ for
    }j\geq0.
  \end{equation*}
  One checks that, for all $j\geq0$, one has $-\beta_j\leq n_j
  \leq\beta_j$. Moreover, using~\eqref{eq:two-a-l}, we get
  \begin{equation*}
    0\leq\beta_j\leq  1+a_{j-1}\cdots a_0\,|n_0|+4a_{j-1},\quad j\ge2.
  \end{equation*}
  Theorem 30 of~\cite{MR98c:11008} implies that, for almost every $a$,
  there exists a subsequence of $(a_j)_j$ that tends to
  $0$. Therefore, we see that, for almost every $a$, for some $j_0$
  sufficiently large, one has $n_{j_0}\in\{-1,0,1\}$. But then, for
  all $j\ge j_0$, $n_{j}\in\{-1,0,1\}$. As $b_j\not \in{\mathcal
    Z}_{a_j}$ $\forall j\ge 0$, the last observation implies that for
  almost any $a$ for all $j$ sufficiently large
  \begin{equation*}
    b_j\in\left\{0,\frac12,\frac{a_j}{2},-\frac{a_j}{2}\right\}.
  \end{equation*}
  Consider the mapping $b\mapsto b_1$, defined by~\eqref{eq:ab}. We
  have
  \begin{equation}
    \label{eq:2}
    \begin{split}
      b_1(0)=
      \begin{cases}
        0\text{ if }\left[\frac1{a_0}\right]\text{ is even}\\
        \frac12\text{ if }\left[\frac1{a_0}\right]\text{ is odd}
      \end{cases},\quad b_1\left(\frac12\right)=-\frac{a_1}2,\\
      b_1\left(\frac{a_0}2\right)= b_1\left(-\frac{a_0}2\right)=
      \begin{cases}
        \frac12\text{ if }\left[\frac1{a_0}\right]\text{ is even}\\
        0\text{ if }\left[\frac1{a_0}\right]\text{ is odd}
      \end{cases}.
    \end{split}
  \end{equation}
  So, for almost all $a$, for all $j$ sufficiently large, one has $
  b_j\in\left\{0,1/2,-a_j/2\right\}$. This completes the proof of
  Lemma~\ref{le:2}.
\end{proof}
%
%
\subsection{Proof of Proposition~\ref{pro:2}}
\label{sec:proof-proposition}
For given $(a_0,b_0)$, define the $(a_n,b_n)$ by
formulas~\eqref{eq:als} and~\eqref{eq:bls}. Recall that for all
$a_0\in(0,1)$ and all $b_0\in{\mathcal B}_{a_0}$, one has $b_j\in
{\mathcal B}_{a_j}$ for all $j\ge 0$.  To prove
Proposition~\ref{pro:2} it is sufficient to prove that, for almost
every $(a_0,b_0)$, there are infinitely many $l$ such that $a_l\leq
\varphi^4(l)$ and $b_l=0$.  The arguments leading to this conclusion
are analogous to the arguments from the end of the
section~\ref{sec:proof-implication--1} (just after the end of proof of
Lemma~\ref{le: infty-of-conditions}). We omit the details and note
only that now we pick $\varphi:\R_+\to\R_+$ so that
\begin{itemize}
\item $\D\sum_{l=1}^{\infty}\varphi^4(l)=+\infty$;
\item $r(x):=\varphi(x)/g(2Ax)$ be a monotonously decreasing function;
\item $\D\lim_{x\to\infty}r(x)=0$;
\item $\varphi(x)\leq 1/2$;
\end{itemize}
where $A$ be the constant defined in~(\ref{eq:4}).\\ \\
As, for all $j\ge 0$, $b_j\in {\mathcal B}_{a_j}$, then to study the
trajectories $\{(a_j,b_j)\subset \R^2,\,j\ge 0\}$ it is possible and
convenient to study trajectories of an one dimensional dynamical
system defined by a piecewise monotonic map of a real interval.  Let
us describe this system.
%
%
%
%
Consider the interval $X=[0,3]$ endowed with the probability measure
$d\nu$ of density (with respect to the Lebesgue measure)
\begin{equation*}
  \nu(x)=\frac1{3\ln 2}\left(\sum_{i=0}^2\frac{1}{x-i+1}
    \car_{[i,i+1]}(x)\right)
\end{equation*}
i.e., up to the factor $1/3$, in each interval $[i,i+1]$, the measure
$\nu$ is the invariant measure for the Gauss map ``shifted'' to
this interval.\\
On $(X,d\nu)$, consider the dynamical system defined by the iterates
of the map $\tilde T:\,\tilde a_0\mapsto \tilde a_1$ such that
\begin{itemize}
\item if $\tilde a_0\in[0,1]$ then
  \begin{equation*}
    \tilde a_1=\left\{\frac1{\tilde a_0}\right\}+
    \begin{cases}
      0&\text{ if }\left[\frac1{\tilde a_0}\right]\text{ if even},\\
      2&\text{ if }\left[\frac1{\tilde a_0}\right]\text{ is odd};
    \end{cases}
  \end{equation*}
\item if $\tilde a_0\in(1,2)$ then
  \begin{equation*}
    \tilde a_1=\left\{\frac1{\tilde a_0-1}\right\}+
    \begin{cases}
      2&\text{ if }\left[\frac1{\tilde a_0-1}\right]\text{ if even},\\
      0&\text{ if }\left[\frac1{\tilde a_0-1}\right]\text{ is odd};
    \end{cases}
  \end{equation*}
\item if $\tilde a_0\in(2,3)$ then
  \begin{equation*}
    \tilde a_1=\left\{\frac1{\tilde a_0-2}\right\}+1.
  \end{equation*}
\end{itemize}
Clearly, for $b_0\in{\mathcal B}_{a_0}$, there is one-to-one
correspondence between the trajectories $\{(a_j,b_j)\subset
\R^2,\,j\ge 0\}$ of the input dynamical system and the trajectories
$\{\tilde a_j\subset \R,\,j\ge 0\}$ of the newly defined one:
\begin{equation}
  a_j=\{\tilde a_j\},\quad \quad
  b_j=\begin{cases}
    0 & \text{ if \ } \tilde a_j\in (0,1)\\
    - a_j/2, & \text{ if \ } \tilde a_j\in (1,2)\\
    1/2, & \text{ if \ } \tilde a_j\in (2,3)
  \end{cases},\quad\quad j\ge 0.
\end{equation}
The value of $b_j$ is coded by $[\tilde a_j]$.\\
Analogously to what was done in section~\ref{sec:almost-sure-growth},
we define
\begin{equation}\label{eq:new-mathfrak-N}
  \mathfrak N(L,\tilde a_0)=\sum_{l=0}^L \chi(
  \sqrt[4]{\tilde a_l}\le\varphi(l)).
\end{equation}
where $\chi(\text{``statement''})$ is equal to 0 if the ``statement''
is false and is equal to 1 otherwise. Recall that $\varphi(l)<1/2$.
Therefore,
\begin{equation*}
  \mathfrak N(L,\tilde a_0)=
  \sum_{l=0}^L \chi(
  \sqrt[4]{a_l}\le\varphi(l))\,\chi(b_l=0).  
\end{equation*}
So, if $ \mathfrak N(L,\tilde a_0)\to \infty$ as $L\to\infty$, then
there are infinitely many $l$ such that $a_l\leq \varphi^4(l)$ and
$b_l=0$.\\
The analysis of the counting function $\mathfrak N$ is similar to that
done when proving Theorem~\ref{thr:2}. We will derive estimates for
appropriate norms of the function $\mathfrak N$. Therefore, we will
use the invariant measure and the exponential mixing of
the dynamical system defined by $\tilde T$. \\
To prove the exponential mixing of the dynamical system defined by
$\tilde T$, we use Theorem 3.1 of~\cite{MR1658635}. We check that
$\tilde T$ defines a weighted covering system (Definition 3.5
of~\cite{MR1658635}). It suffices to prove
\begin{Le}
  \label{le:3}
  Let $P$ be the Perron-Frobenius operator of $\tilde T$.\\
  For any $I\subset X$ non empty open interval, there exists
  $N=N(I)\in\N$ and $C=C(I)>0$ such that $P^N\car_{I}\geq C\,\car_X$.
\end{Le}
\begin{proof}[Proof of Lemma~\ref{le:3}] Recall that the
  Perron-Frobenius operator is defined by the formula
  \begin{equation}
    (Pu)(a_1)=\nu^{-1}(a_1)\,\sum_{a:\,\,\tilde T(a)=a_1}\frac{\nu(a)u(a)}{|\tilde T'(a)|}.
  \end{equation}
  Using the definitions of $\nu$ and $\tilde T$, we get
  \begin{multline*}
    Pu= \left(P_e(u\car_{[0,1]})+ P_o\tau_1(u\car_{[1,2]}) \right)+
    \tau_1^{-1}(P_e+P_o)\tau_2(u\car_{[2,3]})\\+
    \tau_2^{-1}\,\left(P_e \tau_1(u\car_{[1,2]})
      +P_o(u\car_{[0,1]})\right).
  \end{multline*}
  where $\tau_i[u](x)=u(x+i)$ and the operators $P_e$ and $P_o$ are
  acting on $L^1([0,1])$ and defined as
  \begin{gather*}
    (P_eu)(a)=(1+a)\sum_{k\geq1}\frac{u\left((2k+a)^{-1}\right)}
    {(2k+a)(2k+1+a)}\\
    (P_ou)(a)=(1+a)\sum_{k\geq1}\frac{u\left((2k-1+a)^{-1}\right)}
    {(2k-1+a)(2k+a)}.
  \end{gather*}
  Note that $P_e+P_o$ is the Perron-Frobenius operator for the Gauss
  map on $([0,1],d\mu)$ where $d\mu$ is the invariant measure for the
  Gauss map.\\
  Note that, there exists $c>0$ such that
  \begin{itemize}
  \item $P(\car_{[0,1]})\geq c\,(\car_{[0,1]}+\car_{[2,3]})$,
  \item $P(\car_{[1,2]})\geq c\,(\car_{[0,1]}+\car_{[2,3]})$
  \item $P(\car_{[2,3]})\geq c\,\car_{[1,2]}$.
  \end{itemize}
  Hence, one has $P^3(\car_{[i,i+1]})\geq c\,\car_X$ for
  $i\in\{0,1,2\}$. So, it suffices to show that for any interval $I$,
  there exists $i$, $N$ and $c$ so that
  $P^N\car_I\geq c\,\car_{[i,i+1]}$.\\
  For $(n_j)_{j\geq1}$ integers, denote by $[n_1,n_2,\cdots,n_p]$ the
  real number defined by the continued fraction
  \begin{equation*}
    [n_1,n_2,n_3,\cdots,n_p]=\cfrac1{n_1+\cfrac1{n_2+
        \cfrac1{\ddots \cfrac{\ddots}{n_{p-1}+\cfrac1{n_p}}}}}.
  \end{equation*}
  Pick a non-empty open interval $I\subset[0,3]$. It contains an
  interval of the form $[x,x']$ where $x=i+[n_1,\cdots,n_{p-1},n_p]$
  and $x'=i+[n_1,\cdots,n_{p-1},n'_p]$ for some $i\in\{0,1,2\}$ and
  $|n_p-n'_p|=1$. So, it suffices to show Lemma~\ref{le:3} for
  intervals of that form.\\
  Pick now $y=[n_1,\cdots,n_{p-1},n_p]$ and $y'=[n_1,\cdots,n'_p]$
  where $|n_p-n'_p|=1$. By the definition of the Gauss map, one
  obtains
  \begin{equation*}
    P_e\car_{[y,y']}\geq c\car_{[\tilde y',\tilde y]}\text{ if }
    n_1\text{ is even} \quad\text{and}\quad
    P_o\car_{[y,y']}\geq c\car_{[\tilde y',\tilde y]}\text{ if }
    n_1\text{ is odd}
  \end{equation*}
  where $\tilde y=[n_2,\cdots,n_p]$ and $\tilde
  y'=[n_2,\cdots,n'_p]$.\\
  Hence, for $[x,x']$ an interval as above, one gets
  $P\car_{[x,x']}\geq c\car_{i+[\tilde y',\tilde y]}$, where
  $i$ is an index in $\{0,1,2\}$ that depends on $x$.\\
  Applying this $p$ times, we get $\D P^p\car_{[x,x']}\geq
  c\car_{[i,i+1]}$ for some $i\in\{0,1,2\}$. Hence,
  \begin{equation*}
    P^{p+3}\car_I\geq P^{p+3}\car_{[x,x']}\geq c\car_X.
  \end{equation*}
  This completes the proof of Lemma~\ref{le:3}.
\end{proof}
\noindent By Theorem 3.1 of~\cite{MR1658635}, we know that the
dynamical system $(\tilde T,X,d\nu)$ is a covering weighted system
(with a constant weight); hence, it admits a unique invariant measure
and one has exponential mixing estimates for the invariant
measure. Let us now compute the invariant measure for $(\tilde
T,X,d\nu)$. Therefore, we apply $P$ to $\car_X$ and use
$(P_o+P_e)(\car_{[0,1]})=\car_{[0,1]}$ to obtain
\begin{equation*}
  \begin{split}
    P\car_X&= \left(P_e(\car_{[0,1]})+ P_o\tau_1(\car_{[1,2]})
    \right)+
    \tau_1^{-1}(P_e+P_o)\tau_2(\car_{[2,3]})+ \\
    &\hspace{4.5cm}+\tau_2^{-1}\,\left(P_e \tau_1(\car_{[1,2]})
      +P_o(\car_{[0,1]})\right)\\
    &=\car_{[0,1]}+ \tau_1^{-1}\car_{[0,1]}+
    \tau_2^{-1}\car_{[0,1]}\\
    &=\car_X
  \end{split}
\end{equation*}
Hence, the invariant measure of $(\tilde T,X,d\nu)$ has the density
$1$ with
respect to $d\nu$. \\
We now return to the proof of Proposition~\ref{pro:2}. Consider the
function $\mathfrak N(L,\tilde a)$ defined
by~\eqref{eq:new-mathfrak-N}.  To use the same line of reasoning as in
the end of section~\ref{sec:reduct-anlys-dynam}, our goal is to prove
that, when $L\to+\infty$, one has
\begin{equation*}
  \|\mathfrak N(L,\cdot)\|_2=\|\mathfrak
  N(L,\cdot)\|_1(1+o(1)),\quad \|\mathfrak
  N(L,\cdot)\|_1\to\infty,
\end{equation*}
where $\|\cdot\|_1$ and $\|\cdot\|_2$ are the norms of $L^1(X,d\nu)$
and $L^2(X,d\nu)$.  We compute
\begin{gather}
  \label{eq:6N1}
  \|\mathfrak N(L,\cdot)\|_1=\sum_{l=1}^L P(l),\\
  \label{eq:6N2}
  \|\mathfrak N(L,\cdot)\|^2_2=\sum_{l=1}^L P(l)+ 2\sum_{1\leq l<
    m\leq L}P_2(l,m),
\end{gather}
where
\begin{gather*}
  P(l)=\int_0^3\chi(
  \sqrt[4]{\tilde a_l}\le\varphi(l))\,d\nu,\\
  P_2(m,l)= \int_0^3\chi(\sqrt[4]{\tilde a_l}\le\varphi(l))
  \chi(\sqrt[4]{\tilde a_m}\le\varphi(m))\,d\nu.
\end{gather*}
Let us use the results on the dynamical system $(\tilde T,X,d\nu)$ to
derive some useful estimates for $P(l)$ and
$P_2(m,l)$. \\
As the invariant measure of $(\tilde T,X,d\nu)$ has the density $1$
with respect to $d\nu$, we compute
\begin{equation}
  \label{eq:6}
  P(l)=\int_0^3\chi(\sqrt[4]{\tilde a}\le\varphi(l))d\nu
  =\frac1{3\ln 2}\varphi^4(l)(1+O(\varphi^4(l))).
\end{equation}
So,
\begin{equation*}
  \|\mathfrak N(L,\cdot)\|_1=\frac1{3\ln 2}
  \sum_{l=1}^L\varphi^4(l)\left(1+O(\varphi^4(l))\right)
  \vers_{L\to+\infty}+\infty.
\end{equation*}
Exponential mixing (Theorem 3.1 in~\cite{MR1658635}) means that there
exists $C>0$ such that, for all $l<m$, one has
\begin{equation}
  \label{eq:5}
  |P_2(m,l)-P(l)P(m)|\leq CP(l)e^{-(m-l)/C}.
\end{equation}
Under the assumptions made on $\varphi$ at the beginning of
section~\ref{sec:proof-proposition}, using~\eqref{eq:6N1},
\eqref{eq:6N2} and~\eqref{eq:5}, we get
\begin{equation*}
  \|\mathfrak N(L,\cdot)\|_1^2\leq\|\mathfrak
  N(L,\cdot)\|_2^2\leq \|\mathfrak
  N(L,\cdot)\|^2_1+\|\mathfrak N(L,\cdot)\|_1+R_L
\end{equation*}
where
\begin{equation*}
  R_L:=C\sum_{0\leq l<m\leq L}^{L} P(l)e^{-(m-l)/C}=
  O\left(\|\mathfrak N(L,\cdot)\|_1\right).
\end{equation*}
Hence, we obtain that $\|\mathfrak N(L,\cdot)\|_2^2=\|\mathfrak
N(L,\cdot)\|_1^2(1+o(1))$ when $L\to+\infty$. Arguing as in the proof
of Lemma~\ref{le: infty-of-conditions}, we conclude that for almost
every $a=a_0$ and all $b=b_0\in\{0,-a/2,1/2\}$, there exist infinitely
many $l$ such that $a_l\leq\varphi^4(l)$ and $b_l=0$. As we have
already explained, this implies Proposition~\ref{pro:1}.\qed
\def\cprime{$'$} \def\cydot{\leavevmode\raise.4ex\hbox{.}}

\end{document}